\documentclass[12pt]{article}
\usepackage[english]{babel}
\usepackage[utf8]{inputenc}
\usepackage{johd}

\usepackage{amsfonts,amssymb,amsbsy,amsmath,latexsym,graphicx}
\usepackage{mathtools,mathrsfs}
\usepackage{natbib}
\usepackage{amsthm}
\usepackage{dsfont,bbding} 
\usepackage{subfigure,color, booktabs}
\usepackage[margin=1in]{geometry}
\usepackage{placeins}
\usepackage{braket}
\usepackage[title]{appendix}
\graphicspath{{Graphics2/}}

\theoremstyle{definition}
\newtheorem{theorem}{Theorem}

\newtheorem{lemma}{Lemma}
\newtheorem{corollary}{Corollary}
\theoremstyle{definition}

\renewcommand{\phi}{\varphi}
\renewcommand{\epsilon}{\varepsilon}

\newcommand{\uN}{\underline{N}}
\newcommand{\un}{\underline{n}}

\newcommand{\Ntil}{\widetilde{N}}
\newcommand{\Ttil}{\widetilde{T}}
\newcommand{\ttil}{\widetilde{t}}

\newcommand{\Var}[1]{\mathsf{Var}\left[{#1}\right]}

\newcommand{\tnorm}[1]{{\left\vert\kern-0.25ex\left\vert\kern-0.25ex\left\vert #1 \right\vert\kern-0.25ex\right\vert\kern-0.25ex\right\vert}}

\newcommand{\mylabel}[2]{#2\def\@currentlabel{#2}\label{#1}}

\newcommand{\gammahat}{\widehat{\gamma}}

\newcommand{\alphahat}{\widehat{\alpha}}
\newcommand{\betahat}{\widehat{\beta}}

\newcommand{\qhat}{\widehat{q}}
\newcommand{\rhat}{\widehat{r}}

\newcommand{\md}{\mbox{d}}

\newcommand{\Expect}[1]{\mathbb{E}\left[{#1}\right]}
\newcommand{\Prob}[1]{\mathbb{P}\left({#1}\right)}

\newcommand{\cip}{\stackrel{p}{\rightarrow}}

\title{Recruitment prediction for multi-centre clinical trials based on a hierarchical Poisson-gamma  model: asymptotic analysis and improved intervals}

\author{Rachael Mountain$^{1}$, Chris Sherlock$^{1}$ \\\\
        \small $^{1}$Department of Mathematics and Statistics, Lancaster University, UK. \\
}

\date{} 

\begin{document}

\maketitle

\begin{abstract} 
\noindent We analyse predictions of future recruitment to a multi-centre clinical trial based on a maximum-likelihood fitting of a commonly used hierarchical Poisson-Gamma model for recruitments at individual centres. We consider the asymptotic accuracy of quantile predictions in the limit as the number of recruitment centres grows large and find that, in an important sense, the accuracy of the quantiles does not improve as the number of centres increases. When predicting the number of further recruits in an additional time period, the accuracy degrades as the ratio of the additional time to the census time increases, whereas when predicting the amount of additional time to recruit a further $n^+_\bullet$ patients, the accuracy degrades as the ratio of $n^+_\bullet$ to the number recruited up to the census period increases. Our analysis suggests an improved quantile predictor. Simulation studies verify that the predicted pattern holds for typical recruitment scenarios in clinical trials and verify the much improved coverage properties of prediction intervals obtained from our quantile predictor. In the process of  extending the applicability of our methodology, we show that in terms of the accuracy of all integer moments it is always better to approximate the sum of independent gamma random variables by a single gamma random variable matched on the first two moments  than by the moment-matched Gaussian available from the central limit theorem.  \end{abstract}

\noindent\keywords{Asymptotic analysis; Asymptotic correction; Clinical trial recruitment; Multi-centre clinical trial; Poisson process;  Recruitment prediction interval.}

\section{Introduction}
\label{sec.intro}
randomised controlled trials represent the gold standard for evaluating the safety and efficacy of a new healthcare intervention or treatment \citep{gold_standard}. Such trials can require thousands of patients, and so will typically recruit from tens or hundreds of centres. The timely recruitment of patients is widely recognised as a key determinant of the success of a clinical trial \citep{Carter04}. Nonetheless, sources suggest as many as 86\% of all clinical trials fail to reach their required recruitment goals [\citealp{Carlisle}; \citealp{Lamberti}; \citealp{Huang}]. Failure to meet recruitment targets can have numerous negative implications, yet arguably the most critical is inadequate statistical power. In such a scenario, there is an increased risk of type II error, thus potentially preventing or delaying an effective treatment from being approved \citep{Treweek}.

Recruitment of a patient to a clinical trial can be thought of as a three-stage process. Firstly, some  recruitment centres are initiated; more centres can be initiated as the trial progresses. Secondly,  a potential recruit is enroled at a given centre; after a lag, the potential recruit is screened for suitability, and if suitable that patient is randomised onto a particular treatment. Methods for predicting future recruitment usually model the probability of screening success separately, so we focus on the second stage of the process.

Future recruitment is often predicted using deterministic methods, based on the number already recruited up to that time, or historical data \citep*{Carter05}. Such an approach is inadequate due to the stochastic nature of the recruitment process, and a number of stochastic models have been proposed.

\citet{Senn1997} considers a Poisson-based model for a multicentre clinical trial where recruitment follows a Poisson process with a fixed study-wide rate, $\lambda \geq 0$. The time to recruit a given number of patients then follows a gamma distribution. The underlying assumption that recruitment follows a Poisson process is well-accepted in the literature, with many articles exploring an inhomogeneous model with a time-dependent rate \citep{Carter04, Carter05, Tang2012, Lan2019}.

The basic Poisson model outlined above fails to incorporate variation in recruitment rate across centres, as well as the uncertainty in the rate estimate. \citet{AF_2007} propose a random effects model in which recruitment follows a homogeneous Poisson process within each centre, with the centre-specific rates viewed as a sample from a gamma distribution. The time to recruit a given number of patients then follows a Pearson type VI distribution, whilst the number recruited in a given time is negative binomial. This model accounts for staggered centre initiation times and provides a method for predicting recruitment for new centres entering the trial. Citations of \citet{AF_2007} on Google Scholar show that it has also been used by major pharmaceutical companies and in  statistical software to plan drug production and distribution across centres during clinical trials. Further details of the model will be given in Section \ref{sec.setup}.

The Anisimov and Fedorov model (henceforth AF) has been developed and extended in numerous directions. For example, \citet*{Bakhshi} suggest an extra level of hierarchy to incorporate variation from trial to trial in the gamma distribution parameters, with an aim to forecast recruitment for trials yet to begin. 
\citet*{Mijoule} propose a Pareto mixture distribution for the centre rates in place of the gamma. 
Further, \citet{Lan2019} and \citet*{szymon} both incorporate time-varying rates into the AF model, whilst also incorporating parameter uncertainty using the Bayesian paradigm.

Alternative methods have been suggested for modelling patient recruitment outside the Poisson approach, including Monte Carlo simulation \citep*{Abbas}, time series analysis \citep{time_series}, Brownian motions \citep{Lai01, Zhang_Lai}, and a nonparametric approach \citep*{Ying}.

We investigate future predictions based on a maximum likelihood fit of the AF model to multi-centre recruitment data, where a total of $N_\bullet$ patients has been recruited over $C$ centres by a census time, $t$. We then consider two prediction objectives, where prediction intervals are required for either  (1) the total number $N_\bullet^+$ recruited over some additional time $t^+$, or (2) the total time $T^+$ to obtain $n_{\bullet}^+$ additional recruits. In this section, for brevity, we focus on objective (1); similar methods and results are obtained for objective (2).

Within the AF model, the distribution of the predicted number of recruits, $\Ntil_\bullet^+$, has a negative binomial distribution, which depends on the observed data via the  maximum likelihood estimates of the model parameters (MLEs); 
in contrast, the true number recruited, $N_\bullet^+\sim \mbox{Poisson}(\lambda_\bullet t^+)$, where $\lambda_\bullet$ is the sum of the recruitment rates of the individual centres. Let $\qhat_p$ be the $p$th quantile of $\Ntil_\bullet^+$; \emph{i.e.}, the predicted quantile. We first investigate $P_p:=\Prob{N_\bullet^+\le \qhat_p}$ in the limit as $C\rightarrow \infty$, and empirically for finite $C$, and show that the key determinant of the behaviour is the ratio $t^+/t$. The desirable result of $P_p=p$ is only recovered in the limit as $t^+/t\rightarrow 0$, whereas in more typical scenarios $P_p$ can be very different from $p$. The underlying reason for this is that the uncertainty in the MLEs is not being accounted for. Our asymptotic approximation to $P_p$ feeds in to a new methodology which allows us to produce tractable prediction intervals, which have a coverage that is very close to that intended, and with a fraction of the computational cost of any bootstrap-based scheme.  

Our theory, and hence our adjusted interval, is derived under the assumption that  all centres opened at the same time; however, sometimes this is not the case. For example, given a predicted shortfall, perhaps based on our theory, it may be decided to open a new set of centres as well as keeping the existing centres going. Alternatively, or in addition, the existing centres may have been opened at different times. Guided by our theory, we provide an intuitive, tractable methodology for creating a prediction interval in such cases and demonstrate its accuracy in practice via  extensive simulation studies.

Section \ref{sec.setup} describes the AF model in detail, and Section \ref{sec.theory.simult} provides the asymptotic analysis in the case where all centres opened at the same time and details the methodology for creating prediction intervals with almost perfect coverage. Section \ref{sec.theory.difft} describes an empirical extension to this methodology for  situations where the centres opened at different times. Our results and methods are verified via a detailed simulation study in Section \ref{sec.simulate}, and then applied to a clinical-trial recruitment data set in Section \ref{sec.AZdata}. We conclude in Section \ref{sec.discuss} with a discussion. First, however, we define the notations that will be used throughout.  

\subsection{Notations}
Let $C$ be the number of centres, and for $c=1,\dots,C$, let $t_c$ and $N_c$ represent the time for which centre $c$ was open before the census time and number recruited in centre $c$ during the time $t_c$. The shorthand $\boldsymbol{\uN}$ refers to the vector $(N_1,\dots,N_C)$, we let $N_\bullet:=\sum_{c=1}^C N_c$, and when all centres are open for the same time we denote that time by  $t$. For Objective One, let $t^+$ be the additional time ahead at which predictions will be made, and let $N_c^+$ be the number recruited in centre $c$ in that time, with $N^+_\bullet=\sum_{c=1}^C N^+_c$. For Objective Two, let $n^+_\bullet$ be the additional number of recruits sought and let $T^+$ be the additional time taken to recruit this number. 
Table \ref{tab:notations} below summarises these notations, and others that will be introduced later.

\begin{table}[h]
\centering
\caption{Common notations used in this article. Objectives One and Two are abbreviated to O1 and O2, respectively. \vspace{0.5em}}
\label{tab:notations}
\begin{tabular}{cl}
\hline
 $C$ & \# centres \\
 $t$ & (global) census time \\
$t_c$ & time centre $c$ is open before census\\ $N_c$ &\# recruits at centre $c$ at census time\\
$N_{\bullet}$ &$\sum_{c=1}^Cn_{c}$\\
$n_c$& realisation of $N_c$\\
$n_\bullet$& realisation of $N_\bullet$\\
$T^+$& time from census until $n^+_{\bullet}$ new recruits\\
$t^+$& realisation of $T^+$ (O2)
 or specified additional recruitment time after census (O1)\\
 $N_c^+$&\# recruits at centre $c$ during specified time $t^+$\\
$N_{\bullet}^+$&$\sum_{c=1}^CN^+_{c}$\\ $n_{\bullet}^+$& realisation of $N_{\bullet}^+$ (O1) or specified total \# additional recruits required (O2)\\
$\qhat_p$&estimated $p$th quantile for $N^+_\bullet$\\
$\rhat_p$&estimated $p$th quantile for $T^+$\\
\hline
\vspace{1em}
\end{tabular}
\end{table}

The negative binomial distribution of the number of successes until there are $a$ failures when the probability of success is $p$ is denoted $\mathsf{NB}(a,p)$. We use the notation $\cip$ and $\Rightarrow$ to indicate convergence in probability and in distribution, respectively, and $\Phi$ to indicate the cumulative distribution function of a $\mathsf{N}(0,1)$ random variable.

\section{Model and prediction set up}
\label{sec.setup}
\subsection{Model, data and likelihood}
The model assumes that the recruitment rate at centre $c$, for $c=1,\dots,C$, is $\lambda_c$, where 
each $\lambda_c $ is drawn independently from
\begin{equation}
\label{eqn.GamHierarchy}
\lambda_c\sim \mathsf{Gam}(\alpha,\beta).
\end{equation}
Data for centre $c$ are $n_c^1,\dots,n_c^{t_c}$, $n_c:=\sum_{s=1}^{t_c} n_c^s$  and $n_\bullet=\sum_{c=1}^C n_c$. The likelihood for centre $c$ is
\begin{align*}
  L(\alpha,\beta,\theta;n_c^{1:t_c})
  &=
  \int_{0}^\infty
  \frac{\beta^\alpha}{\Gamma(\alpha)}\lambda^{\alpha-1}\exp(-\beta \lambda)
  \prod_{s=1}^{t_c} \frac{\lambda^{n_{c}^s}}{n_c^s!}\exp(-\lambda)
  \md \lambda\\
  &\propto
  \frac{\beta^\alpha}{\Gamma(\alpha)}
  \int_0^\infty \lambda^{\alpha+n_c-1}\exp[-\lambda(\beta+t_c)] \md \lambda\\
  &=
       \frac{\Gamma(\alpha+n_c)}{\Gamma(\alpha)}
  \frac{\beta^\alpha}{(\beta+t_c)^{\alpha+n_c}}.
\end{align*}

Hence, up to an additive constant, the log-likelihood given data $n_c^s$, $s=1,\dots,t$, $c=1,\dots,C$, is
\begin{align}
  \ell(\alpha,\beta)&=
  C\alpha \log \beta
  -\sum_{c=1}^C(\alpha+n_c)\log(\beta+t_c)
  -C\log \Gamma(\alpha)+\sum_{c=1}^C\log \Gamma(\alpha+n_c).
  \label{eqn.NBlike}
  \end{align}
Thus $\boldsymbol{\un}=(n_1,\dots,n_C)$ is a sufficient statistic. In the special case where $t_1=\dots=t_C=t$, the second term in  \eqref{eqn.NBlike} reduces to
$-(C\alpha+n_\bullet)\log(\beta+t)$ and, as we shall see in Lemma \ref{lemma.MLEs}, $\alphahat/\betahat$ depends on $\boldsymbol{\un}$ only through $n_\bullet$.

\subsection{Prediction}
Since $N_c|\lambda_c\sim \mathsf{Po}(\lambda_c t_c)$, given a prior of 
$\mathsf{Gam}(\alphahat,\betahat)$ for $\lambda_c$ and an observation of $n_c$, the posterior distribution for $\lambda_c$ is $\mathsf{Gam}(\alphahat+n_c,\betahat+t_c)$. The distribution of
$\lambda_\bullet:=\sum_{c=1}^C\lambda_c$ is not tractable in general, but in the special case where 
$t_1=\dots=t_C=t$,
$\lambda_\bullet \sim \mathsf{Gam}(C\alphahat+n_\bullet,\betahat+t)$.
In this case, since
$N^+_\bullet|\lambda_\bullet\sim \mbox{Po}(\lambda_\bullet t^+)$, marginalising over $\lambda_\bullet$, the predicted total recruitment in further time $t^+$ is
\begin{align}
      \label{eqn.NBpredictive}
\Ntil_\bullet^+\sim \mathsf{NB}\left(C\alphahat+n_\bullet,\frac{t^+}{\betahat+t+t^+}\right),
\end{align}
which has moments of
\begin{equation}
  \label{eqn.NBmoments}
\Expect{\Ntil_\bullet^+}=\frac{C\alphahat+N_\bullet}{\betahat+t}\times t^+
~~~\mbox{and}~~~
\Var{\Ntil_\bullet^+}=\frac{C\alphahat+N_\bullet}{\betahat+t}\times t^+\times \frac{\betahat+t+t^+}{\betahat+t}.
\end{equation}

Alternatively, if the number of additional recruits is fixed at $n^+_\bullet$ then, $T^+|\lambda_\bullet\sim \mbox{Gam}(n^+_\bullet,\lambda_\bullet)$, so in the case where $t_1=\dots=t_C=t$, the predicted further time $\Ttil^+$ to recruit these has a Pearson VI distribution \citep*[e.g.][]{PVI_distr} with a density of
\begin{equation}
  \label{eqn.PearsonVidens}
f(\ttil^+)=\frac{\Gamma(C\alphahat +n_\bullet+n^+_\bullet)}{\Gamma(C\alphahat+n_\bullet)\Gamma(n^+_\bullet)}
\frac{\betahat^{C\alphahat+n_\bullet}(\ttil^+)^{n^+_\bullet-1}}{(\betahat+t+\ttil^+)^{C\alphahat+n_\bullet+n^+_\bullet}}.
\end{equation}
Thus $\Ttil^+$ has moments of:
\begin{equation}
  \label{eqn.PVImoments}
  \Expect{\Ttil^+}=\frac{(\betahat+t)n^+_\bullet}{C\alphahat+N-1}
  ~~~\mbox{and}~~~
  \Var{\Ttil^+}=
  \Expect{\Ttil^+}\times \frac{(\betahat+t)(C\alpha+N_\bullet+n^+_\bullet-1)}
  {(C\alpha+N_\bullet-1)(C\alpha+N_\bullet-2)}.
\end{equation}

\section{Asymptotic analysis and methodology}
We consider the properties of the quantile estimates under repeated sampling, so that $\boldsymbol{\uN}$ is a random variable, and $\alphahat$ and $\betahat$ are, therefore, random. We examine the probability under the true data-generating mechanism that the quantity of interest, $N_\bullet^+$ or $T^+$, will be less than its predicted quantile. This leads to a tractable formula for an alternative probability, $p^*(p)$, such that $\Prob{N_\bullet^+\le \qhat_{p^*}}\approx p$ or $\Prob{T^+\le \rhat_{p^*}}\approx p$, and hence to prediction intervals with close to the intended coverage. In Section \ref{sec.theory.simult} we consider the scenario where all centres have been open for the same time; an intuitive extension for the more general scenario is given in Section \ref{sec.theory.difft}.

\subsection{All centres opened simultaneously}
\label{sec.theory.simult}
When all centres have been open for the same time, $t$,   $N_\bullet\sim \mathsf{Po}(\lambda_\bullet t)$ is the key (random) summary of the data, instead of $n_\bullet$ for the specific realisation; thus $\alphahat$ and $\betahat$ are random. Importantly, in this case $\alphahat/\betahat$ depends on $\boldsymbol{\uN}$ only through $N_\bullet$.

\begin{lemma}
  \label{lemma.MLEs}
  When $t_1=\dots=t_C=t$, the MLE for the likelihood in \eqref{eqn.NBlike} satisfies $\alphahat/\betahat=N_\bullet/(Ct)$.
  \end{lemma}

\begin{proof}
  Set $\gamma=\alpha/\beta$; from the invariance principle it is sufficient to show that $\gammahat=N_\bullet/(Ct)$. Substituting for $\beta$ and ignoring terms only in $\alpha$,  \eqref{eqn.NBlike} becomes:
  \begin{align*}
  \ell(\alpha,\gamma)&=
  -C\alpha \log \gamma
  -\left(C\alpha + N_\bullet\right)\log(\alpha/\gamma+t)\\
  &=N_\bullet \log \gamma - (C\alpha+N_\bullet)\log(\alpha + \gamma t).
  \end{align*}
  Thus
  \[
  \partial_\gamma \ell=
  \frac{N_\bullet}{\gamma}-\frac{C\alpha+N_\bullet}{\alpha+\gamma t}\times t
  =\frac{\alpha}{\gamma(\alpha+\gamma t)}\left(N_\bullet - \gamma Ct\right),
  \]
  which is zero (and a maximum for $\ell$) when $\gamma=N_\bullet/(Ct)$, as required.
    \end{proof}

We now state our main result.

\begin{theorem}
  \label{thrm.asympt}
  Consider an infinite sequence of recruitment scenarios indexed by the number of recruitment centres, $C=1,2,\dots$. In each scenario, $C$, after each centre has been opened for a fixed common time $t$, $(\alpha, \beta)$ is estimated from data $\boldsymbol{\uN}^{(C)}$ by maximising \eqref{eqn.NBlike}. It is used in \eqref{eqn.NBpredictive} to estimate the $p$th quantile, $q_p^{(C)}$, of the total number, $N_\bullet^{+(C)}$ of recruits in an additional, fixed time $t^+$; it is also used in \eqref{eqn.PearsonVidens} to estimate the $p$th quantile, $r_p^{(C)}$, of the time, $T^{+(C)}$ until a further $n^{+(C)}$ recruits have been obtained. Denoting the quantile estimates as $\qhat_p^{(C)}$ and $\rhat_p^{(C)}$, respectively, the following results hold with $Z\sim N(0,1)$. 
  \begin{enumerate}
    \item
      \[
      \lim_{C\rightarrow \infty}
      \Prob{N_\bullet^{+(C)}\le \qhat_p^{(C)}\mid \boldsymbol{\uN}^{(C)}}\stackrel{D}{=}
      \Phi\left\{
  \sqrt{\frac{t^+}{t}}Z+\Phi^{-1}(p)\sqrt{1+\frac{t^+/\beta}{1+t/\beta}}
      \right\}.
      \]
      However, for large $C$, $\qhat_p^{(C)}/q_p^{(C)} -1=\mathcal{O}(1/\sqrt{C})$; moreover
      \[
      \frac{\qhat_{1-p/2}^{(C)}-\qhat_{p/2}^{(C)}}
      {q_{1-p/2}^{(C)}-q_{p/2}^{(C)}}
      \sim
      \sqrt{1+\frac{t^+/\beta}{1+t/\beta}}.
      \]
    \item 
       If as $C\rightarrow \infty$,  $n_\bullet^{+(C)}/C\rightarrow a>0$,
     \[
\lim_{C\rightarrow \infty}     \Prob{T^{+(C)}\le \rhat_p^{(C)}\mid \boldsymbol{\uN}^{(C)}}\stackrel{D}{=}\
     \Phi\left\{
     \sqrt{\frac{a\beta}{\alpha t}}Z+
     \Phi^{-1}(p)\sqrt{1+\frac{a/\alpha}{1+t/\beta}}
     \right\}.
     \]
      However,  $\rhat_p^{(C)}/r_p^{(C)} -1=\mathcal{O}(1/\sqrt{C})$; moreover
\[
      \frac{\rhat_{1-p/2}^{(C)}-\rhat_{p/2}^{(C)}}
      {r_{1-p/2}^{(C)}-r_{p/2}^{(C)}}
      \sim
      \sqrt{1+\frac{\alpha/a}{1+t/\beta}}.
      \]
\end{enumerate}
\end{theorem}

Theorem \ref{thrm.asympt} is proved in Appendix \ref{sec.proof}. We discuss the consequences for $N_\bullet^+$ in detail; those for $T^+$ are analogous.

Theorem \ref{thrm.asympt} confirms the intuition that the width of any confidence interval estimated using  $(\alphahat,\betahat)$ is wider than that which would be obtained were the total intensity, $\lambda_\bullet$, known precisely; however it also shows that the ratio approaches $1$ as the census time increases.
More importantly, for the median, Theorem \ref{thrm.asympt} suggests that
$\Prob{N_\bullet^+\le \qhat_{0.5}}\approx \Phi(\sqrt{t^+/t}Z)$, so that when $t^+\approx t$, this probability is approximately uniformly distributed on $[0,1]$. By contrast, when $t^+<<t$ the probability concentrates at $\approx 0.5$ as is desirable, and when $t^+>>t$ the probability concentrates around $0$ and $1$ each with a mass of $0.5$, which is not desirable. The theoretical densities for $\Prob{N_\bullet^+\le \qhat_{0.5}}$ as a function of $t$ (with $t^+=400-t)$ are given in  Figure \ref{fig:change_t_m_formula}. For more general quantiles, with $t$ fixed, as $t^+\rightarrow 0$, the probability approaches a point mass at $p$ as desired, but as $t^+\rightarrow \infty$ the same concentration around $0$ and $1$ happens, however, the mass on $1$ is
$\mathbb{P}\left\{Z\ge -\sqrt{t/(\beta+t)} \Phi^{-1}(p)\right\}
=
\Phi\{\sqrt{t/(\beta+t)} \Phi^{-1}(p)\}$.

\begin{figure}[H]
    \centering
    \includegraphics[scale=0.75]{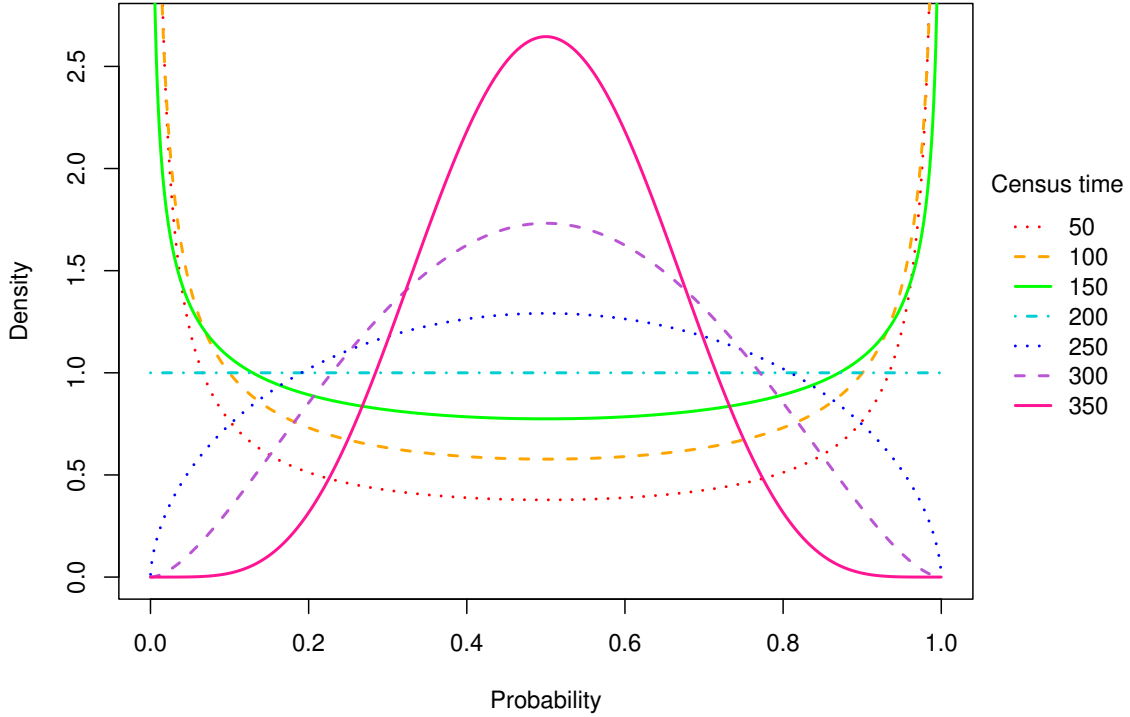}
    \caption{Theoretical density of $\Prob{N_\bullet^+\le \qhat_{0.5}}$ as a function of census time, $t$, with $t^{+}=400-t$, $\alpha=2$, $\beta=150$ and $C=150$.}
    \label{fig:change_t_m_formula}
\end{figure}

Despite this decidedly unintuitive behaviour of the quantile probabilities, Theorem \ref{thrm.asympt} also shows that the relative error in the quantile estimate decays in proportion to $1/\sqrt{C}$ as expected. The resolution of this apparent contradiction lies in the fact that whilst the quantiles for $N_\bullet^+$ and $\Ntil_\bullet^+$  themselves are $\mathcal{O}(C)$, both the discrepancy between them \emph{and} the widths of the distributions are $\mathcal{O}(\sqrt{C})$. The discrepancy between the quantiles also decreases to $0$ as $t^+/t\downarrow 0$, so depending on this ratio the two distributions can closely overlap or almost entirely diverge ($t^+>>t$).

Thus, even though the point estimate of a quantile may be accurate relative to the size of the quantile ($\mathcal{O}(\sqrt{C})$ compared with $\mathcal{O}(C)$), unless $t^+<<t$, prediction intervals will not, in general, provide the intuitive and desirable coverage properties: $\Prob{\qhat_{0.05}\le N_\bullet^+\le \qhat_{0.95}}\approx 0.9$, for example. However, the (asymptotically) correct coverage can be recovered by adjusting the interval, based on Theorem \ref{thrm.asympt}, as we now describe.

Theorem \ref{thrm.asympt} suggests that to obtain a predictive value with the true (asymptotic in $C$) probability $p$ of it not being exceeded, we must target a value $p^*$ such that
      \[
p=
\Expect{      \Phi\left\{
  \sqrt{\frac{t^+}{t}}Z+\Phi^{-1}(p^*)\sqrt{\frac{\beta+t+t^+}{\beta+t}}
      \right\}}.
      \]
Writing $b$ for $\Phi^{-1}(p^*)\sqrt{(\beta+t+t^+)/(\beta+t)}$ and letting $Z'\sim \mathsf{N}(0,1)$ be independent of $Z$, the right hand side may be rewritten as
      \[
      \Prob{Z'\le \sqrt{\frac{t^+}{t}}Z+b}
      =
      \Prob{\sqrt{1+\frac{t^+}{t}}\mathsf{N}(0,1)\le b}
      =\Phi\left(\frac{b}{\sqrt{1+\frac{t^+}{t}}}\right).
      \]
      Rearranging gives
      \[
\sqrt{\frac{t+t^+}{t}}\Phi^{-1}(p)=\sqrt{\frac{\beta+t+t^+}{\beta+t}}\Phi^{-1}(p^*),
\]
so
\begin{equation}
  \label{eqn.aimq}
  p^*=\Phi\left\{
\sqrt{\frac{(\beta+t)(t+t^+)}{t(\beta+t+t^+)}}\Phi^{-1}(p)
  \right\} .
\end{equation}
In practice we do not know $\beta$, and necessarily substitute $\betahat$ for this value. The estimator $\betahat$ is consistent for $\beta$, and so we might expect this approximation to be reasonable. Section \ref{sec.sim.intervals} provides empirical verification that adjustments based on this approximation lead to substantial improvements in coverage.  

\subsection{Different centre opening times}
\label{sec.theory.difft}
We now consider the scenario where $t_1=\dots=t_C$ does not hold. In this case the posterior for $\lambda_\bullet$ is intractable and, hence, so are the distributions for $\Ntil_\bullet^+$ and $\Ttil^+$. Furthermore, Lemma \ref{lemma.MLEs} does not hold.

Although the distribution of $\lambda_\bullet$ is intractable, its moments are not:
\[
\Expect{\lambda_\bullet}
=
\sum_{c=1}^C\frac{\alphahat+n_c}{\betahat+t_c}
~~~
\mbox{and}
~~~
\Var{\lambda_\bullet}
=
\sum_{c=1}^C\frac{\alphahat+n_c}{(\betahat+t_c)^2}.
\]

\begin{theorem}
\label{thrm.mgf}
    Let \(X_1,\dots,X_n\) be independent random variables
    with \(X_i\sim \mathsf{Gam}(\alpha_i,\beta_i)\). Define
    \(S_n=\sum_{i=1}^nX_i\), \(\mu_n=\mathbb{E}[S_n]\) and
    \(\sigma^2_n=\mathsf{Var}[S_n]\). Consider two approximations to \(S_n\)
    obtained by matching the first two moments:
    \begin{align*}
    Z&\sim \mathsf{N}(\mu_n,\sigma^2_n)\\
    G&\sim \mathsf{Gam}(\alpha,\beta),
    \end{align*}
    with \(\alpha/\beta=\mu_n\) and \(\alpha/\beta^2=\sigma^2_n\). Denoting
    the \(j\)the cumulants of \(S_n\), \(Z\) and \(G\) by \(\kappa_j^{S_n}\),
    \(\kappa_j^Z\) and \(\kappa_j^G\), respectively,
    \(0<\kappa_1^{S_n}=\kappa_1^{Z}=\kappa_1^G=\mu_n\) and
    \(0<\kappa_2^{S_n}=\kappa_2^{Z}=\kappa_2^G=\sigma^2_n\) by design, and
    the following holds for all \(j\ge 3\):
    \begin{align*}
    0=\kappa_j^{Z}< \kappa_j^G \le \kappa_j^{S_n}.
    \end{align*}

\end{theorem}

Theorem \ref{thrm.mgf} is proved in Appendix \ref{sec.proof.2}. Since the moment generating function of a random variable is \(M(t)=\exp\{K(t)\}\), where \(K(t)\) is the cumulant generating function, the coefficient of \(t^n\) in \(M(t)\) is a linear combination products of the cumulants
\(\kappa_1,\dots,\kappa_n\) where all coefficients are positive. This
immediately leads to the following:

\begin{corollary}
  \label{corollary}
    With \(S_n\), \(Z\) and \(G\) as defined in Theorem
    \ref{thrm.mgf}, \(\mathbb{E}[Z^j]=\mathbb{E}[G^j]=\mathbb{E}[S_n^j]\) for
    \(j=1,2\), and for all integer \(j\ge 3\),
    \begin{align*}
    \mathbb{E}[Z^j]<\mathbb{E}[G^j]\le \mathbb{E}[S_n^j].
    \end{align*}
\end{corollary}

Theorem \ref{thrm.mgf} and Corollary \ref{corollary} show that a moment-matched gamma approximation to \(S_n\) is, in a sense, strictly better than the moment-matched Gaussian approximation available through the central limit theorem. We, therefore, make the approximation \citep[see also Lemma 2.2 in][]{anisimov_2011} that
\[
\lambda_\bullet
\stackrel{D}{\approx}\lambda^*_\bullet
\sim
\mbox{Gam}(C\alphahat +n^*_\bullet,\betahat+t^*),
\]
where $n^*_\bullet$ and $t^*$ are chosen so that the first two moments of $\lambda^*_\bullet$ match those of $\lambda_\bullet$. Figure \ref{fig:gamma_sum} in Appendix \ref{sec.supp.gamm.app}, and the accompanying text, demonstrate the accuracy of this approximation for two scenarios relevant to trial recruitment that we will describe in Section \ref{sec.sim.intervals}.

The posterior distribution for $\lambda^*_\bullet$ is exactly that which would arise given the $\mbox{Gam}(C\alphahat,\betahat)$ prior if each centre had been open for the same time of $t^*$ and a total of $n_\bullet^*$ patients had been recruited. Thus, if the MLEs from this `data', $\alphahat^*$ and $\betahat^*$ were to satisfy
$\alphahat^*=\alphahat$ and $\betahat^*=\betahat$ then the theory from Section \ref{sec.theory.simult} would follow through exactly. In reality, whatever the partitioning of $n_\bullet^*$ across centres, the data would typically lead to slightly different MLEs $\alphahat^*\ne \alphahat$ and $\betahat^*\ne \betahat$; nevertheless, in the proof of Theorem \ref{thrm.asympt} the most important aspect of the MLEs is their ratio. From Lemma \ref{lemma.MLEs}, 
$\alphahat^*/\betahat^*=n_\bullet^*/Ct^*$, and empirical comparisons of $n_\bullet^*/Ct^*$ against $\alphahat/\betahat$ (see Appendix \ref{sec.supp.gamm.app})  showed a relative error of less than $0.1\%$. 

The methodology for constructing prediction intervals for either $\Ntil_\bullet^+$ or $\Ttil^+$ then proceeds as in Section \ref{sec.theory.simult}, using $\alphahat$ and $\betahat$ under the assumption that $\lambda_\bullet\equiv\lambda^*_\bullet$.

\section{Empirical verification of theory and methodology}
\label{sec.simulate}
Simulations were carried out to test the asymptotic theory and methods proposed in this paper for finite numbers of centres, $C$. A large number (20000 unless otherwise stated) of realisations of the parameters $\lambda_1,\dots,\lambda_C$, and hence the sample $(n_{1} , \ldots , n_{C})$ were simulated for a given set of parameter values. For each realisation, the parameters $\alpha$ and $\beta$ were estimated using maximum likelihood and the quantile of interest, $q_p$ or $r_p$ was estimated. Either $\Prob{N_\bullet^+\le \qhat_p}$ or $\Prob{T^+\le \rhat_p}$ was then calculated exactly using the known (simulated) $\lambda_1,\dots,\lambda_C$. The results outlined below will primarily focus on predicting $N_\bullet^+$. 

Unless specified otherwise, the following parameter values were used: $\alpha = 2$, $\beta = 150$, $C = 150$, $t=200$. The latter two values are the defaults used when considering varying census times and centre numbers respectively.

When predicting $N_\bullet^+$, the total trial length was set to $\tau=t+t^+=400$, since 
with the default $C$, $\Expect{N_\bullet+N_\bullet^+}=C(\alpha/\beta)(t+t^+)=800$, a reasonable size for a Phase III clinical trial. Furthermore, the census time $t$ was chosen from $\boldsymbol{\mathbb{T}_{1}} = \{50, 100, 150, 200, 250, 300, 350\}$ and the number of centres, $C$, was chosen from $\boldsymbol{\mathbb{C}_{1}} = \{20, 50, 100, 150, 200, 250, 300, 400\}$. When examining predictions of $T^+$ we fixed $n^+_\bullet=200$ and selected
$t\in \boldsymbol{\mathbb{T}_{2}} = \{50, 100, 150, 200, 300,\\ 500, 1000\}$ and
$C\in \boldsymbol{\mathbb{C}_{2}} = \{20, 50, 100, 150, 200, 300, 500, 1000\}$.

When conducting simulations with varying number of centres, we set $\beta=C$ to maintain a fixed expected number of recruits per unit time. In Appendix \ref{sec.addn.simthrm} we explore an alternative scenario where $\beta=150$ is fixed as $C$ varies.

\subsection{Verification of Theorem \ref{thrm.asympt}}
\label{sec.sim.thrm}

\begin{figure}[h]
    \centering
    \begin{tabular}{cc}
    \includegraphics[width = 0.48 \textwidth, height = 7.5cm]{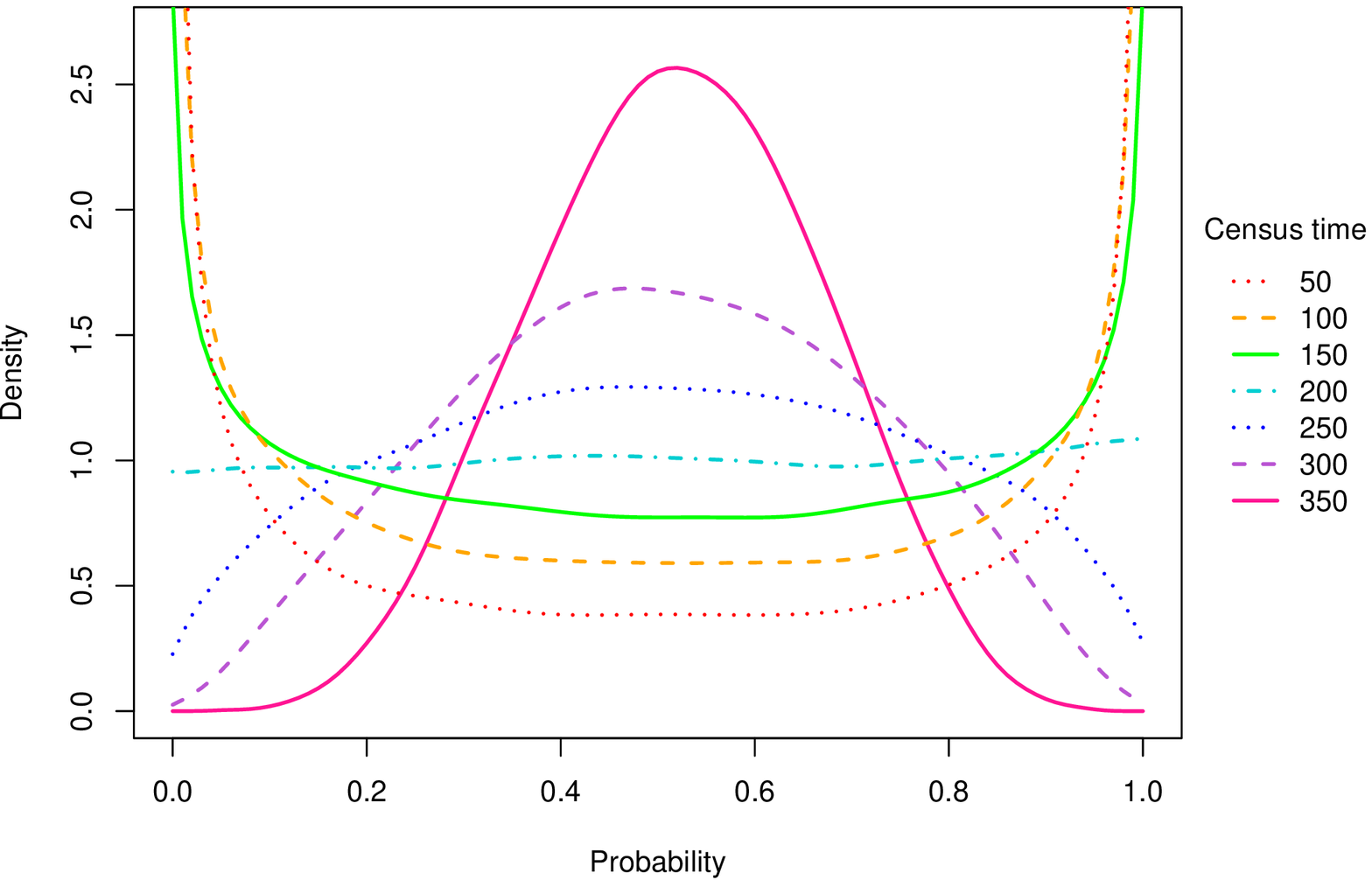} & \includegraphics[width = 0.48 \textwidth, height = 7.5cm]{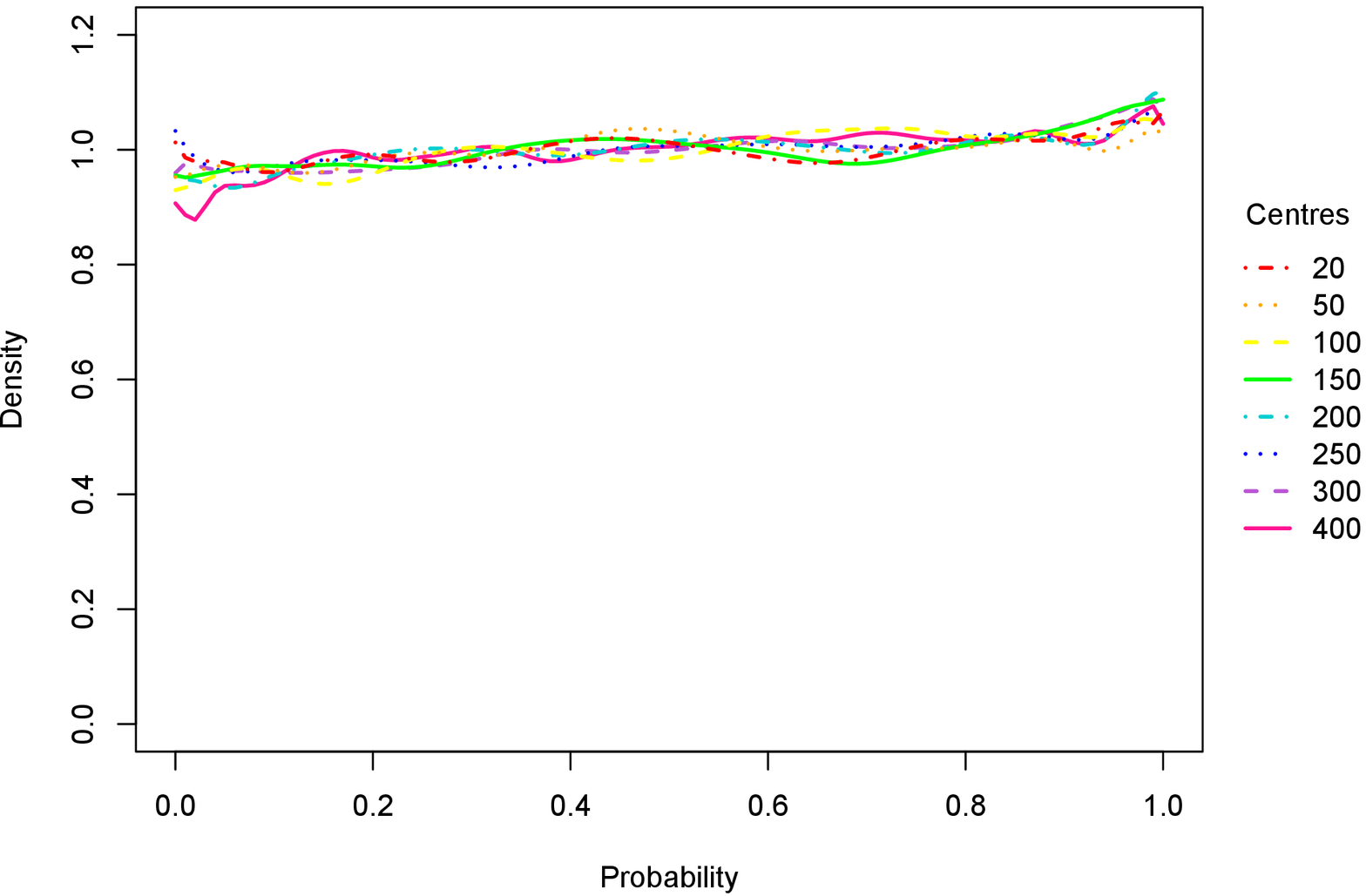}
    \end{tabular}
    \vspace{1em}
    \caption{Estimated density (over repeated sampling) of $\Prob{N_\bullet^+\le \qhat_{0.5}}$ for each $t \in \boldsymbol{\mathbb{T}_{1}}$ with $t^{+}=400-t$ (left) and for each $C \in \boldsymbol{\mathbb{C}_{1}}$ with $t=t^{+}=200$ and $\beta = C$ (right).}
    \label{fig:change_t_m}
\end{figure}

Figure \ref{fig:change_t_m} shows the empirical distribution of $\Prob{N_\bullet^+\le \qhat_{p}}$ over repeated simulation and, hence, estimates $\qhat_p$, for the median, $p=0.5$. The left panel varies the census times $t\in \boldsymbol{\mathbb{T}_1}$, whilst the right panel fixes $t$ (and hence $t^+=\tau-t$) and varies the number of centres, $C\in \boldsymbol{\mathbb{C}_1}$. The shape of the density function for $\Prob{N_\bullet^+\le \qhat_p}$ depends on the ratio of $t^+/t$ and shows very little variation with $C$, just as described in Section \ref{sec.theory.simult}, and matching almost perfectly  the relevant theoretical curves in Figure \ref{fig:change_t_m_formula}. In particular, when $t=t^+$, as in all cases in the right panel, the distribution is very close to uniform, empirically verifying the, perhaps unintuitive, result that increasing the number of centres in the trial, thus increasing the sample size upon which the MLEs are based, does not affect the accuracy of the quantile estimates. The theory predicts that the lines in the right panel should be horizontal; however, there is a slight positive gradient. This is because the theory is based on a continuous approximation whereas $N_\bullet^+$ is a discrete random variable. The density function for $\mathbb{P}(N_\bullet^+< \widehat{q}_p)$ (not shown) exhibits a slight negative gradient, supporting this explanation.

Figure \ref{fig:change_t_25} repeats Figure \ref{fig:change_t_m_formula} and the left panel of Figure \ref{fig:change_t_m} but for the $p=0.25$ quantile. Again, the empirical results match the theory almost perfectly. As with $p=0.5$, the estimate improves with increasing census time, but as predicted in Section \ref{sec.theory.simult}, when $t \ll t^{+}$, the mass is now not evenly distributed between the regions close to 0 and close to 1. 

\begin{figure}[h]
    \centering
    \includegraphics[scale = 0.53, trim={0 0.55cm 0 0},clip]{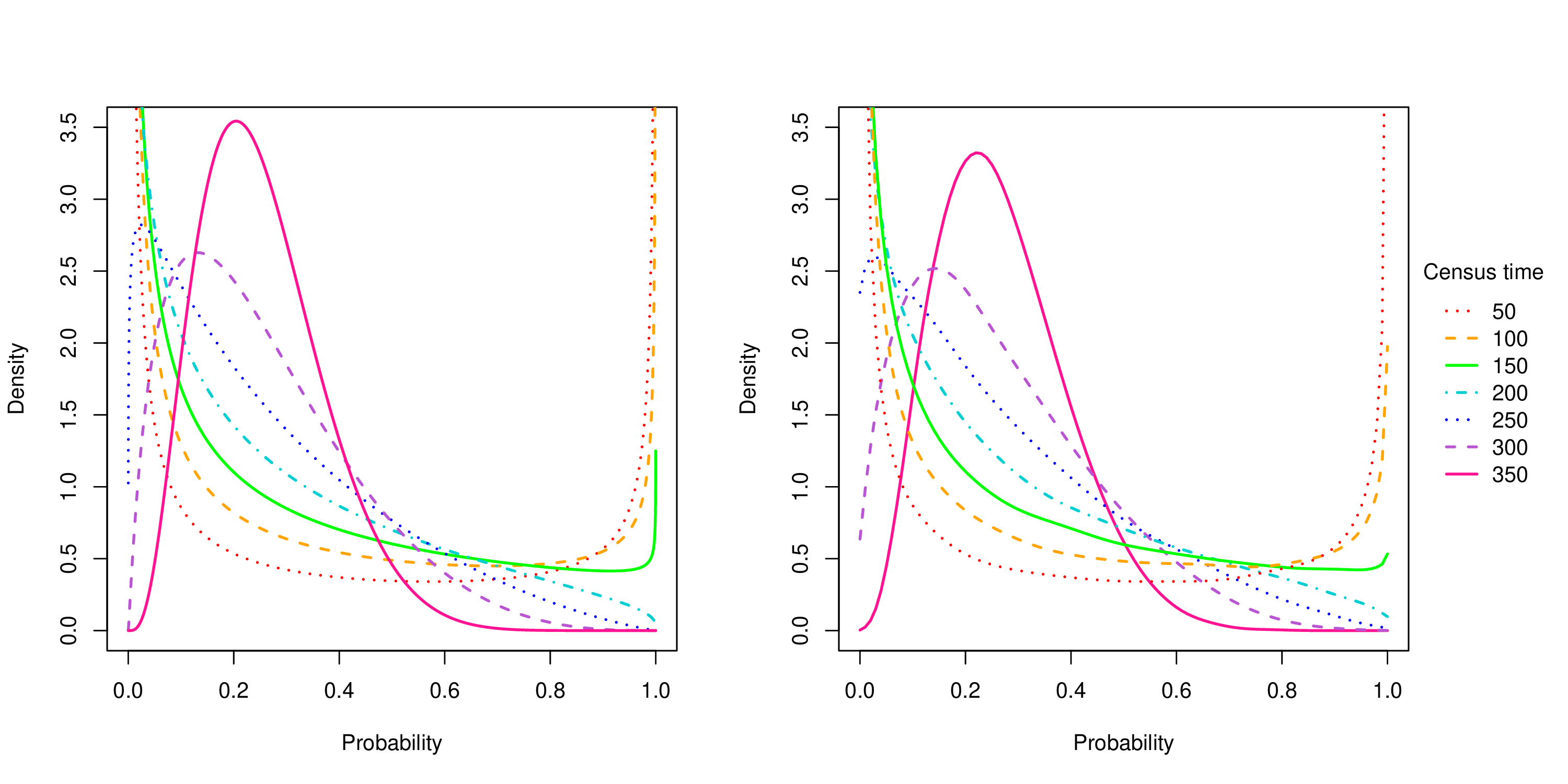} 
    \caption{Theoretical density (left) and estimated density over repeated sampling (right) of $\Prob{N^{+}_\bullet \leq \hat{q}_{0.25}}$ for each $t \in \boldsymbol{\mathbb{T}_{1}}$, with $t^{+}=400-t$.}
    \label{fig:change_t_25}
\end{figure}

When predicting quantiles for $T^+$, Theorem \ref{thrm.asympt} suggests that the accuracy of the quantile is primarily dependent on the ratio of $n_\bullet^+/n_\bullet$. Thus with a fixed $n_\bullet^+$ and $t$, and with $\beta=C$, there is essentially no change in the prediction accuracy; Figure \ref{fig:PVI_C} captures the close agreement between the theoretical predictions and empirical results in this case. 

\begin{figure}[H]
    \centering
        \includegraphics[scale = 0.53]{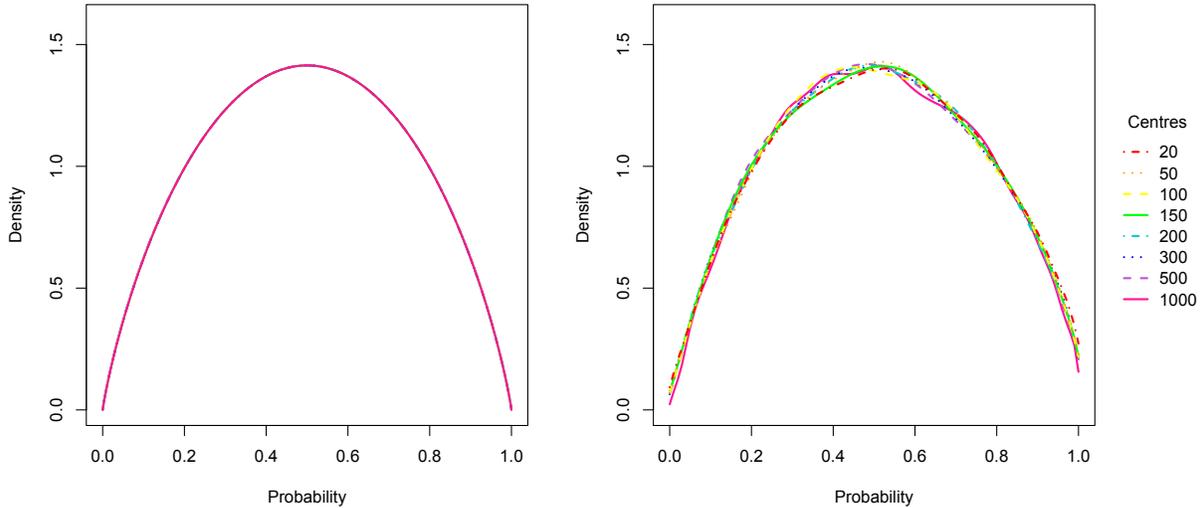}
    \caption{Theoretical density (left) and estimated density over repeated sampling (right) of $\Prob{T^{+} \leq \rhat_{0.5}}$ for each $C \in \boldsymbol{\mathbb{C}_{2}}$ with $n^{+}_\bullet=200$ fixed across all simulation runs. In the left-hand figure, all eight curves coincide.}
    \label{fig:PVI_C}
\end{figure}

Alternatively, with $\beta$ fixed, as the number of centres increases and/or as the census time increases, since each of these increases $n$, the  density curves for $\Prob{T^+\le \rhat_p}$ transition from a concentration at 0 and 1 to a point mass at $p$ (see Appendix \ref{sec.addn.simthrm} for theoretical predictions and empirical verification).  For further validation of Theorem \ref{thrm.asympt}, plots for $p=0.25$ with $t$ varying are also provided in Appendix \ref{sec.addn.simthrm}.

\subsection{Adjusted prediction intervals}
\label{sec.sim.intervals}

We now study empirically the effectiveness of using quantiles based on $p^*(p)$ to derive prediction intervals, and compare with intervals based directly on $p$.  At each simulation, a naive, unadjusted 90\% interval was estimated by calculating $\hat{q}_{p}$ for
$p=0.05$ and $p=0.95$. An adjusted $90\%$ interval was also derived by using $p^*(p)$ from \eqref{eqn.aimq} instead of $p$, both for $p=0.05$ and $p=0.95$.  The performance of the intervals was assessed for each method by calculating 
the mean, over 2000 simulations, of the true prediction interval coverage. The mean width of the prediction intervals was also recorded. We first consider the case were all centres opened simultaneously, then the case of different centre opening times.

\textbf{All centres opened simultaneously.} Table \ref{tab:quantile_adj_ninety} shows the results for each $t \in \boldsymbol{\mathbb{T}_1}$, and $t^{+} = \tau - t$. The unadjusted method gives satisfactory results for $t \ll t^{+}$ only, as is to be expected given Theorem 1. For all other scenarios, the quantiles are inaccurately estimated and the coverage can be far less than intended, as low as 63.7\% for a census time early on in the trial. Further diagnostics showed approximately equal contributions to undercoverage from $\qhat_{0.05}$ being too high and $\qhat_{0.95}$ being too low. In contrast,  by applying \eqref{eqn.aimq}, the coverage is consistently improved upon and corrected to almost exactly the desired 90\%. The improved coverage does come with a cost of an increased interval width, but the increase seems proportionate. 

\begin{table}[H]
\centering
\caption{The mean (over repeated sampling) of the true coverage probability and width of an intended 90\% prediction interval for $N^{+}_\bullet$ using the unadjusted and adjusted methods.\vspace{1em}}
\label{tab:quantile_adj_ninety}
\begin{tabular}{ccccc}
\hline
 & \multicolumn{2}{c}{\textbf{Unadjusted}} & \multicolumn{2}{c}{\textbf{Adjusted}} \\ \hline
 & Coverage  (\%) & $w$ & Coverage  (\%) & $w$ \\ \hline
$t=50, t^{+}=350$ & 63.7 & 140.5 & 89.1 & 245.6 \\ \hline
$t=100, t^{+}=300$ & 76.3 & 118.2 & 89.5 & 160.9 \\ \hline
$t=150, t^{+}=250$ & 81.9 & 99.0 & 89.5 & 120.0 \\ \hline
\multicolumn{1}{l}{$t=200, t^{+}=200$} & 84.9 & 82.2 & 89.6 & 92.9 \\ \hline
\multicolumn{1}{l}{$t=250, t^{+}=150$} & 86.9 & 66.6 & 89.8 & 72.0 \\ \hline
\multicolumn{1}{l}{$t=300, t^{+}=100$} & 88.2 & 51.3 & 89.8 & 53.6 \\ \hline
\multicolumn{1}{l}{$t=350, t^{+}=50$} & 89.2 & 34.5 & 89.9 & 35.1 \\ \hline
\end{tabular}
\end{table}

When $\beta=0$, \eqref{eqn.aimq} gives $p^*=p$: no correction is needed. We, therefore, also examined the effect of our adjustment when data are simulated using a much lower true parameter value, $\beta=50$. In this case, the lowest coverage was $77.8\%$, observed when $(t,t^+)=(50,350)$, improving to $90.2\%$ after our adjustment, whilst when $t=t^+=200$ the coverage improved from $84.1\%$ to $90.0\%$; the full tabulation is provided in Appendix \ref{sec.addn.sim.intervals}.

Similar improvements to those in Table \ref{tab:quantile_adj_ninety}, but for the $95\%$ prediction interval are also provided in Appendix \ref{sec.addn.sim.intervals}, confirming that the $p^{*}$ adjustment performs equally well when adjusting quantiles which are further into the tails of the distribution. A further table in Appendix \ref{sec.addn.sim.intervals} demonstrates an even more striking improvements than in Table \ref{tab:quantile_adj_ninety}, found when creating a $90\%$ predictive interval but with $C=20$; for example, when $(t,t^+)=(50,350)$ the coverage improved from $59.2\%$ to $89.7\%$.

\textbf{Different centre opening times.} 
We consider two different opening time scenarios: (1) the centre opening times are drawn uniformly and independently from the interval $[0,t]$, and (2) half of the centres are opened at time $0$ and half of the centres open at time $t$. The former mimics a gradual coming online of new centres, whilst the latter scenario could occur when an initial interim analysis suggests that many new centres must be opened to achieve the required sample size.

The investigation into quantile adjustment to obtain a 90\% prediction interval (Table \ref{tab:quantile_adj_ninety}) was repeated for opening time scenarios (1) and (2), and the results are provided in Tables \ref{tab:quantile_adj_diffc} and \ref{tab:halfhalft}, respectively. The prediction intervals for these cases were constructed according to the methodology of Section \ref{sec.theory.difft}. Additional diagnostics for the moment matching were also recorded: the mean (over repeated samples) of $t^*$, the ratio of this to the mean (over repeated samples) of the mean (over centres) of the $t_c$'s, and the ratio of the mean of the $n^*_\bullet$ to the mean of the $n_\bullet$. 

In both cases, the intervals obtained by combining the methodology proposed in Section \ref{sec.theory.difft} with \eqref{eqn.aimq} produce coverages very close to $90\%$, whatever the census time. By contrast the unadjusted intervals suffered from coverages as low as $48\%$ when $t=50$. Typically the values of $t^*$ and $n_\bullet^*$ are lower than $t$ and $n_\bullet$ (although their ratio is almost unchanged; see Section \ref{sec.theory.difft}), representing the increased uncertainty in parameter values because some centres have not been open for the full time interval. The especially poor coverage of the unadjusted intervals results because it is now the ratio $t^{+}/t^{*}$ that determines the extent of the undercoverage.

\begin{table}[h]
\centering
\caption{The mean (over repeated sampling) of the true coverage probability and width  of an intended 90\% prediction interval for $N^{+}_\bullet$ using the unadjusted and adjusted methods for opening time scenario (1). \vspace{1em}}
\label{tab:quantile_adj_diffc}
\begin{tabular}{cccccccc}
\hline
 & \multicolumn{1}{c}{} & \multicolumn{1}{c}{} & \multicolumn{1}{c}{} & \multicolumn{2}{c}{\textbf{Unadjusted}} & \multicolumn{2}{c}{\textbf{Adjusted}} \\ \hline
 & \multicolumn{1}{c}{$t^{*}$} & \multicolumn{1}{c}{$t^{*} / t_{c}$} & \multicolumn{1}{c}{$n_{\bullet}^{*} / n_{\bullet} $} & Coverage (\%) & $w$ & Coverage (\%) & $w$ \\ \hline
$t=50, t^{+}=350$ & 24.4 & 0.957 & 0.956 & 49.3 & 143.1 & 89.2 & 341.4 \\ \hline
$t=100, t^{+}=300$ & 46.5 & 0.921 & 0.920 & 65.0 & 125.3 & 89.6 & 220.3 \\ \hline
$t=150, t^{+}=250$ & 67.2 & 0.891 & 0.890 & 72.7 & 106.7 & 89.6 & 160.0 \\ \hline
\multicolumn{1}{l}{$t=200, t^{+}=200$} & 86.9 & 0.866 & 0.865 & 77.6 & 88.8 & 89.7 & 119.7 \\ \hline
\multicolumn{1}{l}{$t=250, t^{+}=150$} & 105.9 & 0.845 & 0.843 & 81.3 & 71.5 & 89.7 & 88.7 \\ \hline
\multicolumn{1}{l}{$t=300, t^{+}=100$} & 124.2 & 0.826 & 0.825 & 84.2 & 54.3 & 89.7 & 62.6 \\ \hline
\multicolumn{1}{l}{$t=350, t^{+}=50$} & 142.1 & 0.810 & 0.809 & 87.1 & 35.5 & 89.8 & 38.2 \\ \hline
\end{tabular}
\end{table}

\begin{table}[h]
\centering
\caption{The mean (over repeated sampling) of the true coverage probability and width  of an intended 90\% prediction interval for $N_\bullet^{+}$ using the unadjusted and adjusted method for opening time scenario (2). \vspace{1em}}
\label{tab:halfhalft}
\begin{tabular}{cccccccc}
\hline
 & \multicolumn{1}{l}{} & \multicolumn{1}{l}{} & \multicolumn{1}{l}{} & \multicolumn{2}{c}{\textbf{Unadjusted}} & \multicolumn{2}{c}{\textbf{Adjusted}} \\ \hline
 & \multicolumn{1}{c}{$t^{*}$} & \multicolumn{1}{c}{$t^{*} / t_{c}$} & \multicolumn{1}{c}{$n_{\bullet}^{*} / n_{\bullet} $} & Coverage (\%) & $w$ & Coverage (\%) & $w$ \\ \hline
$t=50, t^{+}=350$ & 21.7 & 0.867 & 0.863 & 48.1 & 145.1 & 89.1 & 360.4 \\ \hline
$t=100, t^{+}=300$ & 38.3 & 0.766 & 0.763 & 60.0 & 126.8 & 89.1 & 240.0 \\ \hline
$t=150, t^{+}=250$ & 51.2 & 0.683 & 0.679 & 66.7 & 108.7 & 89.0 & 179.0 \\ \hline
\multicolumn{1}{l}{$t=200, t^{+}=200$} & 61.7 & 0.612 & 0.614 & 71.1 & 90.9 & 88.9 & 136.0 \\ \hline
\multicolumn{1}{l}{$t=250, t^{+}=150$} & 70.1 & 0.561 & 0.558 & 75.3 & 73.4 & 89.0 & 101.3 \\ \hline
\multicolumn{1}{l}{$t=300, t^{+}=100$} & 77.0 & 0.513 & 0.511 & 79.6 & 55.6 & 89.4 & 70.8 \\ \hline
\multicolumn{1}{l}{$t=350, t^{+}=50$} & 82.8 & 0.473 & 0.471 & 84.2 & 36.2 & 89.6 & 41.8 \\ \hline
\end{tabular}
\end{table}

Equivalent tables for $T^+$ for opening time scenarios (1) and (2), presented in Appendix \ref{sec.addn.sim.intervals}, show similar dramatic improvements.

\FloatBarrier

\subsection{Application to clinical trial recruitment data}
\label{sec.AZdata}
Finally, we applied our methodology to recruitment data from an oncology clinical trial. The recruitment centres opened at different times, thus the methodology of Section \ref{sec.theory.difft} applies. For anonymisation reasons, all times in the data set were jittered by up to a week, and in the plot described below both time and cumulative recruitment have been rescaled to lie in the interval $[0,1]$.

We examined the 41 centres that had opened by time 0.125 and calculated 90\% prediction intervals for the total recruitment from these centres for the remainder of the recruitment period. With a single data set it is impossible to obtain true coverage probabilities, however, we can compare the predicted intervals with the true number recruited. Figure \ref{fig:AZ_data} shows the prediction intervals as dashed lines in red (unadjusted intervals) and dotted lines in blue (adjusted intervals) with the actual recruitment numbers shown as a solid black line. The recruitment goes outside of the unadjusted interval just before time 0.6 yet remains entirely within the adjusted prediction interval.

A diagnostic likelihood-ratio test \cite[see][]{szymon} with a null hypothesis  that the Poisson process is   time-homogeneous (as assumed by the model)  produced p-values of 0.39 (data up to the census time) and 0.50 (all data). Diagnostic Q-Q plots (see Appendix \ref{sec.addn.diagnostics}) suggest that the assumption of a gamma distribution in \eqref{eqn.GamHierarchy} is reasonable.

\begin{figure}[H]
    \centering
    \includegraphics[scale=0.55]{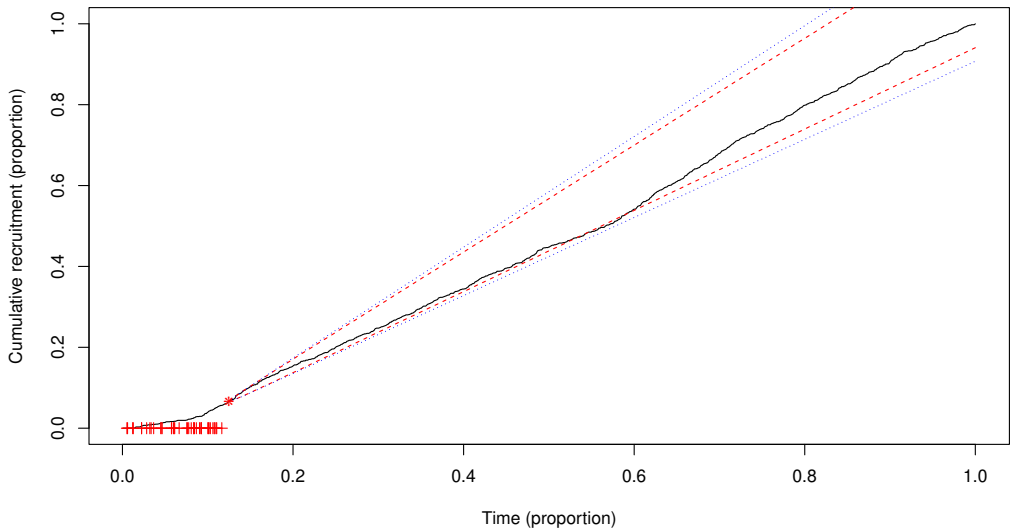}
    \caption{Recruitment (black, solid) for an oncology clinical trial. The estimated 90\% prediction intervals are shown by red dashed lines (unadjusted intervals) and dotted blue lines (adjusted intervals), and the ``$+$" symbols indicate centre opening times.}
    \label{fig:AZ_data}
\end{figure}

\section{Discussion}
\label{sec.discuss}
{We must start by pointing out that the model described in Section \ref{sec.setup} is just that: a model. The hierarchical nature allows the borrowing of information from centres that have been open for some time and enables sensible predictions for newly opened centres, the Poisson process is a reasonable first approximation for the recruitment process at an individual centre, and the gamma hierarchical distribution is chosen for tractability. The model does not account for the myriad logistical issues that might occur during a trial, affecting recruitment, and even were this not the case, data do not arise from the model. However, the model has gained  traction in the industry and has been developed further by a number of authors (see the introduction of this article).}

{Theorem \ref{thrm.asympt} first provides insight into when prediction intervals obtained by simply plugging in the  parameter point  estimates might be adequate; for example, when the future time horizon is small compared with the time for which the trial has been running. However, often the future time horizon is at least as long as the current length of the trial, and in this situation  the coverage of plug-in intervals is poor. The methodology resulting from Theorem \ref{thrm.asympt}, essentially, takes parameter uncertainty into account to create prediction intervals with almost exactly the intended coverage. } 

{Alternatives that allow parameter uncertainty to inform prediction intervals include Bayesian methods,  \cite[e.g.][]{szymon, zhang_long_2010} which are typically computationally expensive, or the bootstrap, which is usually even more expensive. Our method has the same cost as the standard plug-in, frequentist approach.}

{The diagnostics detailed at the end of Section \ref{sec.AZdata} suggest that the model of \cite{AF_2007} is suitable for the oncology data set which we examine, but this might not always be the case. The simulation study in Appendix \ref{gam.mix} suggests robustness to moderate departures from the hierarchical gamma distribution and robustness of improvements to the prediction intervals through our method. However, as demonstrated, for example, in \cite{szymon}, if the intensity curve for each centre is strongly time-dependent, predictions based on the assumption of a homogeneous Poisson process can be wildly inaccurate, and a time-inhomogeneous Poisson process might be more appropriate \citep[e.g.][]{Lan2019,szymon}. If the exact form of the time-inhomogeneity  is known then the standard  time transformation used for the Poisson process ($t\rightarrow \lambda_c \int_0^t  a(s) ds$, where $\lambda_c a(s)$ is intensity at true time $s$ for centre $c$), with one transformed time scale per centre, permits the application of our correction to predicting the number of new recruits in a given additional (true) time. However, the time-dependency  typically contains unknown parameters, and our correction as it stands cannot account for the uncertainty in these. Future research could look into extending our method to allow for this.}

{Theorem \ref{thrm.asympt}, upon which our prediction adjustment is based, describes the limit as the number of centres $C\rightarrow \infty$. Our simulations suggest that the approximation based on the limit result works well even when $C$ is as low as 20; however, it is unlikely to hold for very low centre numbers. Furthermore, experience has shown that for very low centre numbers it is possible for the likelihood to increase monotonically as $\alpha\rightarrow \infty$ and $\beta\rightarrow \infty$ with $\alpha/\beta$ fixed (this can occur when the counts for individual centres are under-dispersed). Relevant historical data might then be brought in to make parameter estimation more robust; however $C$ would still be low and the intended coverage might not be achieved.}

{This article has considered scenarios where centres can open at different times up until the census time, additional centres may be opened at the census time (perhaps driven by the results of the analysis) and we wish to predict the total recruitment for these centres into the future. A more general opening time scenario would also allow for centres coming online at different times after the census time. This could be incorporated into predictions of recruitment over the remainder of the recruitment period via a more general definition of $\lambda_\bullet$, which would become a weighted sum of the individual intensities, with a centre's weight being the fraction of the future time that it would be open for. How to deal with the converse problem in this scenario: predicting the time to recruit a certain number of patients, is an open problem.}

{Prediction using the model of \cite{AF_2007} relies on the true centre opening times, which are rarely known in
advance. There is often a plan and a back up plan, however, and it is straightforward \citep[see the appendix of][]{szymon} to combine the \cite{AF_2007} model with a standard survival model for the opening time of each centre conditional on the planned opening time and, potentially, other covariates. Alternatively, \cite{Lan2019} models centre opening times as realisations from an inhomogeneous Poisson process. With either of these approaches, once the model has been fitted using the data up to the census time, it is straightforward to repeatedly simulate sets of future opening times. One would then obtain a mixture of negative-binomial distributions for the distribution of the number of additional recruits over additional time $t^+$. The mixture could be approximated by a single negative-binomial distribution and our method applied directly to that. This would mainly be an extension of the model of \citet{AF_2007}, and would certainly be interesting to explore; the aim of this paper, however, is to analyse the existing method of \citet{AF_2007} and, in addition to new insights on performance, provide improved prediction intervals.}

\section*{Acknowledgements}

The first author acknowledges support from award: NIHR-MS-2016-03-01 Lancaster University.\vspace*{-8pt}

\bibliographystyle{abbrvnat}
\bibliography{bib}

\newpage

\appendix
\appendixpage
\addappheadtotoc
\counterwithin{figure}{section}
\counterwithin{table}{section}

The proofs of Theorems \ref{thrm.asympt} and \ref{thrm.mgf} are given in Appendices \ref{sec.proof} and \ref{sec.proof.2} respectively. Appendix \ref{sec.supp.gamm.app} provides empirical evidence of the accuracy of the approximations used in Section \ref{sec.theory.difft} of the main article. Appendices \ref{sec.addn.simthrm}, \ref{sec.addn.sim.intervals} and \ref{sec.addn.diagnostics} contain additional material respectively for Sections \ref{sec.sim.thrm} \ref{sec.sim.intervals} and \ref{sec.AZdata} of the main article. Finally, Appendix \ref{gam.mix} investigates the robustness of the adjusted intervals to departures from the hierarchical gamma distribution used in the model.

\section{Proof of Theorem \ref{thrm.asympt}}
\label{sec.proof}
In this section, since all quantities except $\alpha$, $\beta$, $t$ and $t^+$ are functions of $C$, we suppress this superscript; further, since no random variables correspond to an individual centre (they are all totals) we suppress the subscript `$_\bullet$',  altering $\lambda_\bullet^{(C)}$ to $\lambda$, $N_\bullet^{(C)}$ to $N$, $N_\bullet^{+(C)}$ to $N^+$ and $\Ntil_\bullet^{+(C)}$ to $\Ntil^+$.
Further, since $\lambda \sim \mathsf{Gam}(C\alpha,\beta)$ and $N|\lambda \sim \mathsf{Po}(\lambda t)$, Chebyshev's inequality gives: $\lambda t/C\cip t\alpha/\beta$ and  $N/(\lambda t)\cip 1$, and hence $N/C\cip t\alpha/\beta$. Finally, by the Central Limit Theorem (CLT):
\begin{align}
  \label{eqn.CLTN}
Z':=  (N-\lambda t)/\sqrt{\lambda t}\Rightarrow Z&\sim \mathsf{N}(0,1).
\end{align}
We prove Parts 1 and 2 of the theorem separately. In each case we initially condition on the random variable $(\lambda,\boldsymbol{\uN})$; however, in the limit, the  probability of interest depends on this random variable only through $Z'\Rightarrow Z$.
\subsection{Proof of Part 1}
Combining Lemma \ref{lemma.MLEs} with \eqref{eqn.NBmoments} gives
\begin{equation}
  \label{eqn.NBmomentsN}
\Expect{\Ntil^+\mid \boldsymbol{\uN}}=\frac{Nt^+}{t}
~~~
\mbox{and}
~~~
\Var{\Ntil^+\mid \boldsymbol{\uN}}=\frac{Nt^+}{t}
\times \frac{\betahat+t+t^+}{\betahat+t}.
\end{equation}
Moreover, \eqref{eqn.NBmomentsN} gives
\begin{equation}
\label{eqn.NBvarlimit}
  \frac{\Var{\Ntil^+\mid \boldsymbol{\uN}}}{\Var{N^+\mid \lambda}}
=
\frac{N}{\lambda t}\times \frac{\betahat+t+t^+}{\betahat+t}
\cip
\frac{\beta+t+t^+}{\beta+t},
\end{equation}
since $\betahat\cip \beta$ by the asymptotic consistency of the MLE.

Conditional on $N$, let 
$\Ntil_c^+\sim \mathsf{NB}\left(\alphahat+N/C,\frac{t^+}{\betahat+t+t^+}\right)$ be independent. Then $\Ntil^+\stackrel{D}{=}\sum_{c=1}^C\Ntil^+_c$. Also $N^+|\lambda ~\sim \mathsf{Po}(\lambda t^+)$, so as $C\rightarrow \infty$, which implies $N/C\cip t\alpha/\beta$, the CLT gives 
\begin{align}
  \label{eqn.CLTNp}
  (N^+-\lambda t^+)/\sqrt{\lambda t^+}\mid \lambda&\Rightarrow \mathsf{N}(0,1),\\
  \label{eqn.CLTNtilp}
  (\Ntil^+-\Expect{\Ntil^+\mid \boldsymbol{\uN}})/\sqrt{\Var{\Ntil^+\mid \boldsymbol{\uN}}}\mid \boldsymbol{\uN}&\Rightarrow \mathsf{N}(0,1).
\end{align}

Substituting \eqref{eqn.CLTN} into \eqref{eqn.NBmomentsN}
\[
\frac{\Expect{\Ntil^+\mid \boldsymbol{\uN}}-\Expect{N^+\mid \lambda}}{\sqrt{\Var{N^+\mid \boldsymbol{\uN}}}}
=
\frac{\frac{t^+}{t}(\lambda t + \sqrt{\lambda t}Z')-\lambda t^+}{\sqrt{\lambda t^+}}
\Rightarrow
\sqrt{\frac{t^+}{t}}Z.
\]
Incorporating this with \eqref{eqn.CLTNtilp} and
\eqref{eqn.NBvarlimit}, the prediction of the $p$th quantile, $\qhat_p$, satisfies
\begin{align}
  \frac{\qhat_p-\Expect{N^+\mid \lambda}}{\sqrt{\Var{N^+\mid \lambda}}}
  ~\mid~N
  &\cip
  \frac{\Expect{\Ntil^+\mid \boldsymbol{\uN}}-\Expect{N^+\mid \lambda}+\Phi^{-1}(p)\sqrt{\Var{\Ntil^+\mid \boldsymbol{\uN}}}}{\sqrt{\Var{N^+\mid \lambda}}}\nonumber\\
  \label{eqn.qasymp}
  &\Rightarrow
\sqrt{\frac{t^+}{t}}  Z+\Phi^{-1}(p)\sqrt{\frac{\beta+t+t^+}{\beta+t}}.
\end{align}
From \eqref{eqn.CLTNp} and \eqref{eqn.qasymp}, the probability the true realisation is less than the predicted quantile is
\begin{align*}
  \Prob{N^+\le \qhat_p\mid \boldsymbol{\uN},\lambda}\cip
  \Phi\left(\frac{\qhat_p-\Expect{N^+\mid \lambda}}{\sqrt{\Var{N^+\mid \lambda}}}\right)
  \Rightarrow
  \Phi\left(\sqrt{\frac{t^+}{t}}Z+\Phi^{-1}(p)\sqrt{\frac{\beta+t+t^+}{\beta+t}}\right).
\end{align*}
Since this does not depend on $\lambda$,  it is also the limit of $\Prob{N^+\le \qhat_p\mid \boldsymbol{\uN}}$,
as required. Furthermore, from  \eqref{eqn.qasymp} and \eqref{eqn.CLTNp}, the  discrepancy between the quantile approximation and the true quantile satisfies
\[
\frac{\qhat_p-q_p}{\sqrt{\Var{N^+\mid \lambda}}}
\approx
\sqrt{\frac{t^+}{t}}Z + \Phi^{-1}(p)\left[\sqrt{\frac{\beta+t+t^+}{\beta+t}} -1\right].
\]
Since   $q_p/C\rightarrow \alpha/\beta$ and $\Var{N^+\mid \lambda}=\mathcal{O}(C)$, the relative discrepancy is $\mathcal{O}(1/\sqrt{C})$. Finally, the same data are used to estimate $q_{1-p/2}$ and $q_p$ so the value of $Z$ is the same, and  from \eqref{eqn.qasymp}, 
\[
\frac{\qhat_{1-p/2}-\qhat_{p/2} }{\sqrt{\Var{N^+\mid \lambda}}}
\approx
\left[ \Phi^{-1}(1-p/2)-\Phi^{-1}(p/2)\right]\sqrt{\frac{\beta+t+t^+}{\beta+t}}.
\]
The expression for the relative widths of the estimated and true confidence intervals then follows from \eqref{eqn.CLTNp}.

\subsection{Proof of Part 2}
Firstly, since $T^+|\lambda\sim \mbox{Gam}(n^+,\lambda)$,
\begin{equation}
  \label{eqn.EVarTplus}
  \Expect{T^+\mid \lambda}=\frac{n^+}{\lambda}=\frac{n^+/C}{\lambda/C}
  \cip \frac{a\beta}{\alpha}
  ~~~\mbox{and}~~~
  C\Var{T^+\mid \lambda}=C\frac{n^+}{\lambda^2}\cip \frac{a \beta^2}{\alpha^2}.
\end{equation}
Combining Lemma \ref{lemma.MLEs} with \eqref{eqn.PVImoments}and using the asymptotic consistency of the MLEs, 
\begin{align}
\label{eqn.ETpgN}
  \Expect{\Ttil^+\mid \boldsymbol{\uN}}&
=
\frac{(\betahat+t)n^+}{C\alphahat + N -1}
=
\frac{(\betahat+t)n^+t}{(\betahat+t)N-t}\cip \frac{a\beta}{\alpha}.\\
  \nonumber
  C\Var{\Ttil^+\mid \boldsymbol{\uN}}&=
  \Expect{\Ttil^+\mid \boldsymbol{\uN}}
  \times
  \frac{(\betahat+t)(\alphahat+N/C+n^+/C-1/C)}
       {(\alphahat+N/C-1/C)(\alphahat+N/C-2/C)}\\
\nonumber
       &\cip
       \frac{a\beta}{\alpha}
       \times \frac{(\beta+t)(\alpha+\alpha t/\beta+a)}{(\alpha+\alpha t/\beta)^2}\\
       \label{eqn.VarTplustil}
       &=
       \frac{a\beta^2(1+a/\alpha+t/\beta)}{\alpha^2(1+t/\beta)}.
       \nonumber
\end{align}
Thus
\begin{equation}
\label{eqn.var.T.ratio}
\frac{\Var{\Ttil^+\mid \boldsymbol{\uN}}}{\Var{T^+\mid \lambda}}
\cip
\frac{1+a/\alpha+t/\beta}{1+t/\beta}.
\end{equation}
Also, from the second equality in \eqref{eqn.ETpgN},
\begin{align}
\frac{\Expect{\Ttil^+\mid \boldsymbol{\uN}}-\Expect{T^+\mid \lambda}}{\sqrt{\Var{T^+\mid \lambda}}}
&=
\nonumber
\frac{1}{\sqrt{n^+/\lambda^2}}
\times
\left[\frac{(\betahat+t)n^+t}{(\betahat+t)N-t}-\frac{n^+}{\lambda}\right]\\
&=
\nonumber
\sqrt{n^+}\left[
\frac{(\betahat+t) (\lambda t-N)+t}{(\betahat+t)N-t}
\right]\\
&=
\nonumber
\sqrt{\frac{n^+/C}{\lambda t/C}}\left[
\frac{(\betahat+t) (\lambda t-N)/\sqrt{\lambda t}+t/\sqrt{\lambda t}}{(\betahat+t)N/(\lambda t)-t/(\lambda t)}
\right]\\
&\Rightarrow
-\sqrt{\frac{a\beta}{\alpha t}}
\times
Z,
\label{eqn.TtilEdiff}
\end{align}
by \eqref{eqn.CLTN}.
Now $\lambda T^+\sim \mathsf{Gam}(n^+,1)=\sum_{i=1}^{n^+}E_i$, where
the $E_i\sim \mathsf{Exp}(1)$ are independent and identically distributed, so the central limit theorem gives
\[
\frac{T^+-\Expect{T^+\mid \lambda}}{\sqrt{\Var{T^+\mid \lambda}}}
=
\frac{\lambda T^+-\Expect{\lambda T^+\mid \lambda}}{\sqrt{\Var{\lambda T^+\mid \lambda}}}
\Rightarrow
\mathsf{N}(0,1).
\]
Further, $\Ttil^+\mid \boldsymbol{\uN}\stackrel{D}{=}\mbox{Gam}(n^+,1)/\lambda\mid \boldsymbol{\uN}=G_1/G_2$, where $G_1\sim \mbox{Gam}(n^+,1)$ and $G_2\sim \mbox{Gam}(C\alphahat,\betahat)$ are independent.  Since  $n^+\rightarrow \infty$ as $C\rightarrow \infty$ and the MLEs are consistent, the delta method and the CLT give: 
${(\Ttil^+-\Expect{\Ttil^+\mid \boldsymbol{\uN}})}/\sqrt{\Var{\Ttil^+\mid \boldsymbol{\uN}}}~\mid \boldsymbol{\uN}
\Rightarrow
\mathsf{N}(0,1)$.
Hence,
\begin{align*}
  \Prob{T^+\le \rhat_p\mid \boldsymbol{\uN},\lambda}
  &
  \cip \Phi\left(\frac{\Expect{T}+\Phi^{-1}(p)\sqrt{\Var{\Ttil^+}}-\Expect{T^+}}{\sqrt{\Var{T^+}}}\right)\\
  &
  \Rightarrow
  \Phi\left(
  -\sqrt{\frac{a\beta}{\alpha t}} Z
+
\Phi^{-1}(p)\sqrt{\frac{1+a/\alpha+t/\beta}{1+t/\beta}}
  \right).
\end{align*}
As with the proof of Part 1, this does not depend on $\lambda$ so is also the limit of $\Prob{T^+\le r_p\mid \boldsymbol{\uN}}$.
Finally, from \eqref{eqn.EVarTplus}, \eqref{eqn.var.T.ratio}, \eqref{eqn.TtilEdiff} and the two CLT applications above,
\[
\frac{\rhat_p-r_p}{\sqrt{\Var{T^+\mid \lambda}}}
\approx
-\sqrt{\frac{a\beta}{\alpha t}}
\times
Z+\Phi^{-1}(p)\left\{\sqrt{\frac{1+t/\beta+a/\alpha}{1+t/\beta}}-1\right\}
\]
Since $\Var{T^+\mid \lambda}=\mathcal{O}(a/C)$ and $r_p=\mathcal{O}(a)$ the second part follows. The expression for the relative widths of the estimated and true confidence intervals follows analogously to the proof for Part 1.

\section{Proof of Theorem \ref{thrm.mgf}}
\label{sec.proof.2}
Since the cumulant
generating function for \(Z\) is \(K_Z(t)=\mu_nt+\sigma_n^2t^2/2\),
\(\kappa_j^Z=0\) for all \(j\ge 3\), so we consider the other two sets
of cumulants. The cumulant generating function for \(G\) is

\begin{align*}
K_G(t)&=-\alpha \log(1-t/\beta),
\end{align*}

so \(\kappa_j^G=(j-1)! \alpha/\beta^j>0\). From this, the \(j\)th
cumulant of \(S_n\) is

\begin{align*}
\kappa_j^{S_n}&=(j-1)!
\sum_{i=1}^n \frac{\alpha_i}{\beta_i^j}>0.
\end{align*}

The matched moments for \(G\) give
\(\alpha/\beta=\sum_{i=1}^n\alpha_i/\beta_i\) and
\(\alpha/\beta^2=\sum_{i=1}^n\alpha_i/\beta_i^2\). Thus

\begin{align*}
\frac{\kappa^G_j}{(j-1)!}
&=
\frac{\alpha}{\beta^j} = \frac{\alpha}{\beta}\times \frac{1}{\beta^{j-1}}\\
&=
\sum_{i=1}^n \alpha_i/\beta_i
\times
\frac{\left(\sum_{i=1}^n\alpha_i/\beta_i^2\right)^{j-1}}
{\left(\sum_{i=1}^n\alpha_i/\beta_i\right)^{j-1}}\\
&=
\frac{\left(\sum_{i=1}^n \alpha_i/\beta_i^2\right)^{j-1}}
{\left(\sum_{i=1}^n\alpha_i/\beta_i\right)^{j-2}}.
\end{align*}

Write \(c_i=\alpha_i/\beta_i\) and \(d_i=\alpha_i/\beta_i^2\), and
imagine that \(c_i\) and \(d_i\) are realisations from random variables
\(C\) and \(D\), where each possible value has a probability of \(1/n\),
then we have

\begin{align*}
\frac{\kappa^{S_n}_j}{n(j-1)!}
&=
\mathbb{E}\left[\frac{D^{j-1}}{C^{j-2}}\right], \\
\frac{\kappa^{G}_j}{n(j-1)!}
&=
\frac{\mathbb{E}[D]^{j-1}}{\mathbb{E}[C]^{j-2}}.
\end{align*}

Let \(\mathcal{H}_j\) be the statement
``\(\mathbb{E}[D^{j-1}/C^{j-2}]\ge \mathbb{E}[D]^{j-1}/\mathbb{E}[C]^{j-2}\).''
If \(\mathcal{H}_j\) is true then \(\kappa_j^{S_n}\ge \kappa^G_j\). To
prove \(\mathcal{H}_j\) for all \(j\ge 2\) it is sufficient to redefine
\(C\leftarrow C/\mathbb{E}[C]\) and \(D\leftarrow D/\mathbb{E}[D]\) so
\(\mathbb{E}[C]=\mathbb{E}[D]=1\) and prove

\begin{align*}
\mathcal{H}'_j:~\mathbb{E}\left[\frac{D^{j-1}}{C^{j-2}}\right]\ge 1.
\end{align*}

Now \(\mathcal{H}'_2\) is true trivially. For \(\mathcal{H}'_j\) with
\(j\ge 3\) we apply the Cauchy-Schwarz inequality and tackle odd and
even \(j\) separately.

If \(j\ge 3\) is odd, from \(\mathcal{H}'_{j/2+1/2}\),

\begin{align*}
1&\le
\mathbb{E}\left[\frac{D^{(j-1)/2}}{C^{(j-1)/2-1}}\right]^2
=
\mathbb{E}\left[\frac{D^{(j-1)/2}}{C^{j/2-1}}\times C^{1/2}\right]^2\\
&\le
\mathbb{E}\left[\frac{D^{j-1}}{C^{j-2}}\right]
\mathbb{E}\left[C\right]
=
\mathbb{E}\left[\frac{D^{j-1}}{C^{j-2}}\right].
\end{align*}

If \(j\ge 4\) is even, from \(\mathcal{H'}_{j/2+1}\)

\begin{align*}
1&\le
\mathbb{E}\left[\frac{D^{j/2}}{C^{j/2-1}}\right]^2
=
\mathbb{E}\left[\frac{D^{(j-1)/2}}{C^{j/2-1}}\times D^{1/2}\right]^2\\
&\le
\mathbb{E}\left[\frac{D^{j-1}}{C^{j-2}}\right]
\mathbb{E}\left[D\right]
=
\mathbb{E}\left[\frac{D^{j-1}}{C^{j-2}}\right].
\end{align*}

Thus, by induction \(\mathcal{H}'_j\) holds for all \(j\ge 2\), and so
does \(\mathcal{H}_j\). Hence \(\kappa_j^{S_n}\ge \kappa_j^{G}\) for all
\(j \ge 2\). Indeed, by design we have equality for \(j=1,2\).

\section{Supporting information  for Section \ref{sec.theory.difft}}
\label{sec.supp.gamm.app}

Figures \ref{fig:gamma_sum} and \ref{fig:abhat_nct} support the use of the moment matching approximation  proposed in Section \ref{sec.theory.difft}. Figure \ref{fig:gamma_sum} shows the accuracy of the moment-matched gamma approximation to the distribution of $\lambda_{\bullet}$, as well as a CLT-based Gaussian approximation,   using rates arising from opening time scenario (1). The moment-matched gamma performs very well, and is superior to the CLT for small numbers of centres, while both are very accurate for large $C$. The Gaussian approximation is purely present for comparison, since a gamma distribution is required for tractability of the integrals over $\lambda_\bullet$, both for $\Ntil_\bullet^+$ and $\Ttil^+$. Plots for opening time scenario (2) (not included) show a similarly good fit. Figure \ref{fig:abhat_nct} provides an empirical comparison of $\hat{\alpha}/\hat{\beta}$ against $n^{*}/Ct^{*}$ for the two opening time scenarios. The plots support the use of the MLEs from the original data for the hypothetical data set where $n^{*}_{\bullet}$ patients have been recruited in time $t^{*}$, as outlined in Section \ref{sec.theory.difft} of the main article.

\begin{figure}[ht]
    \centering
    \begin{tabular}{cc}
         \includegraphics[width = 0.47 \textwidth]{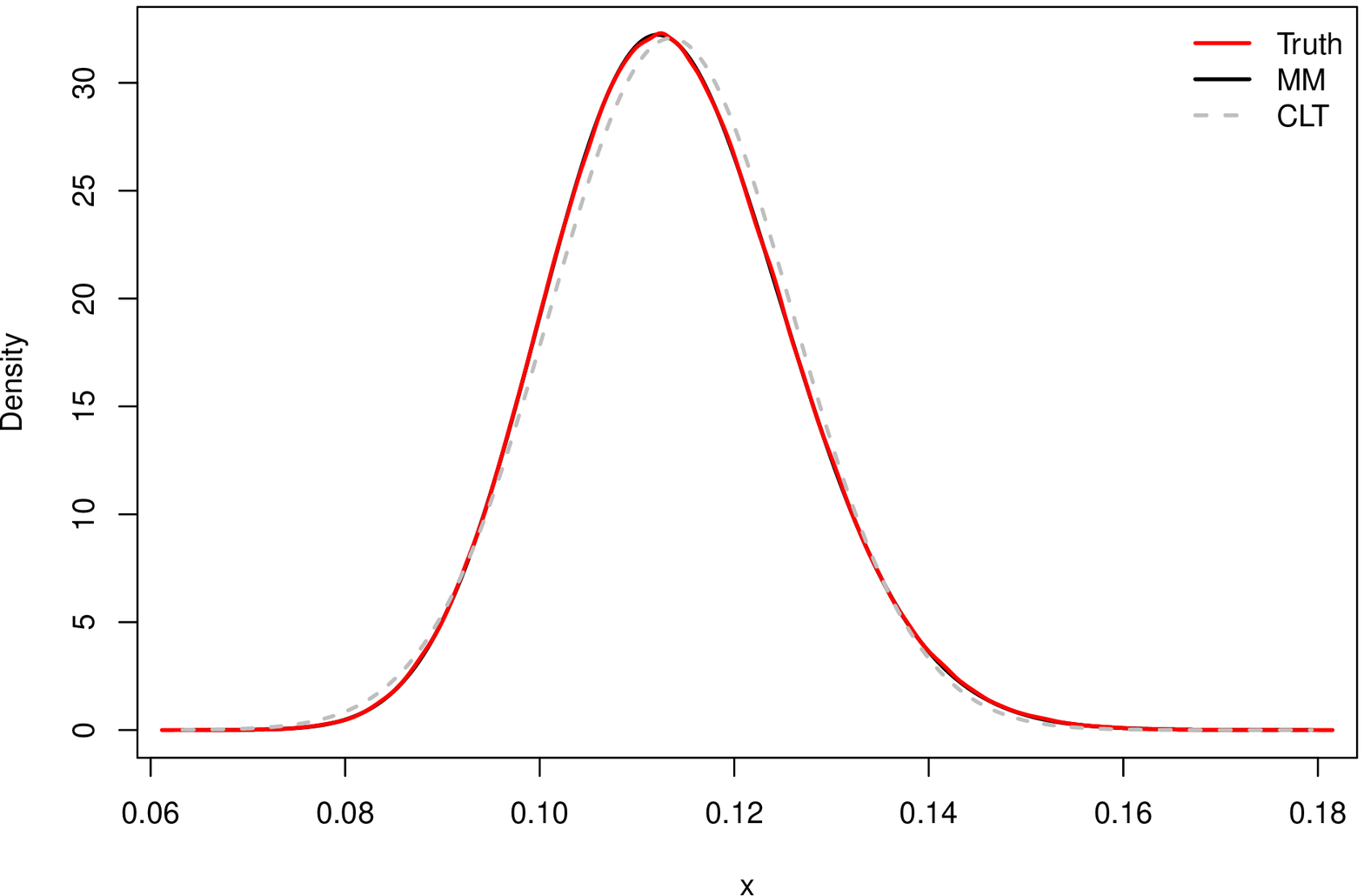} & \includegraphics[width = 0.47 \textwidth]{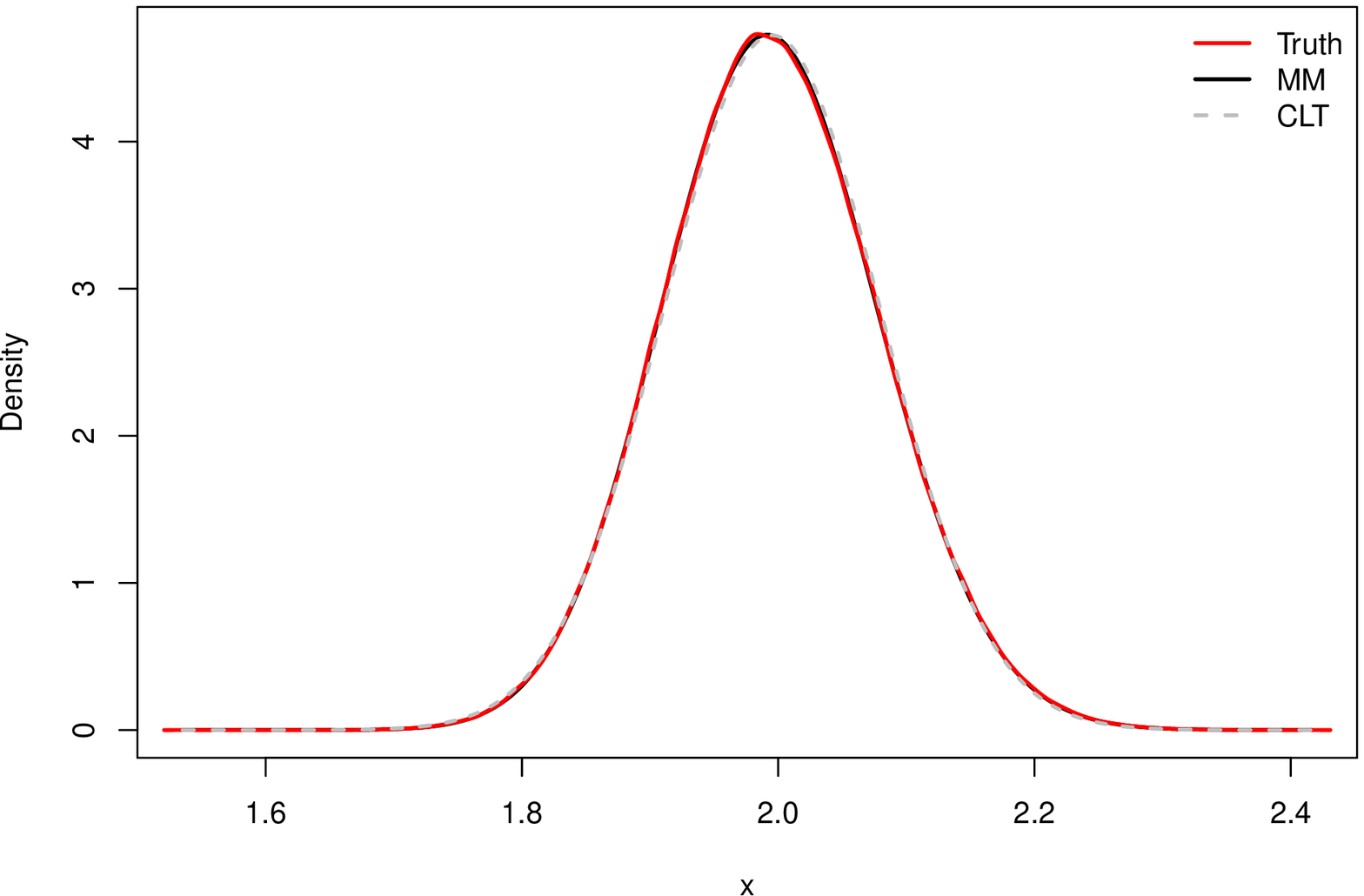}  
    \end{tabular}
    \caption{Comparison of using the moment-matched (MM) method and the central limit theorem (CLT) to estimate the sum of gamma random variables with different rate parameters, from opening time scenario (1), with $C=20$ (left) / $C=150$ (right), $\alpha = 2$, $\beta = 150$ and $t=200$.}
    \label{fig:gamma_sum}
\end{figure}
\FloatBarrier
\vspace{2em}

\begin{figure}[ht]
    \centering
    \begin{tabular}{cc}
        \includegraphics[width = 0.47 \textwidth]{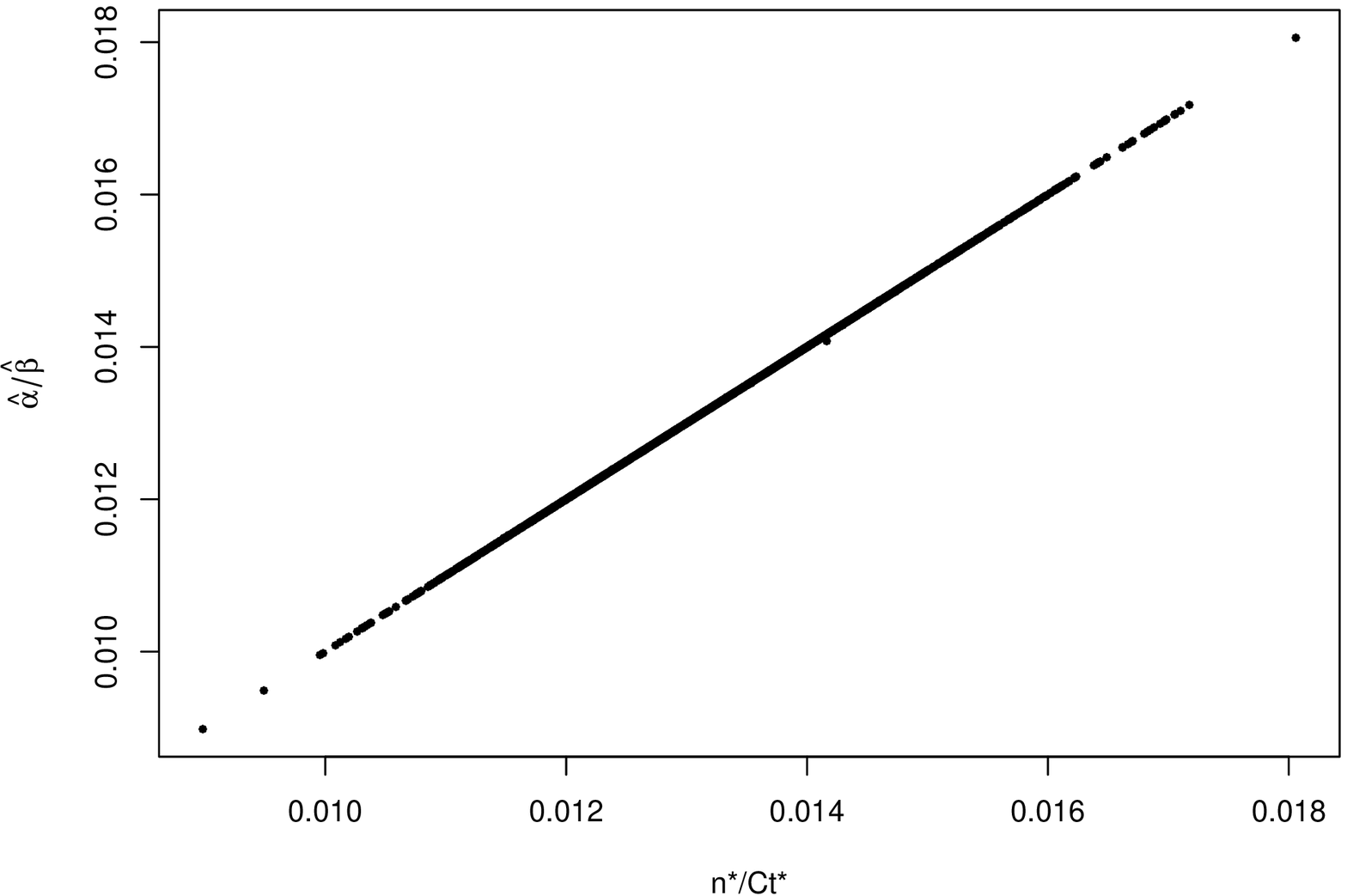} &
        \includegraphics[width = 0.47 \textwidth]{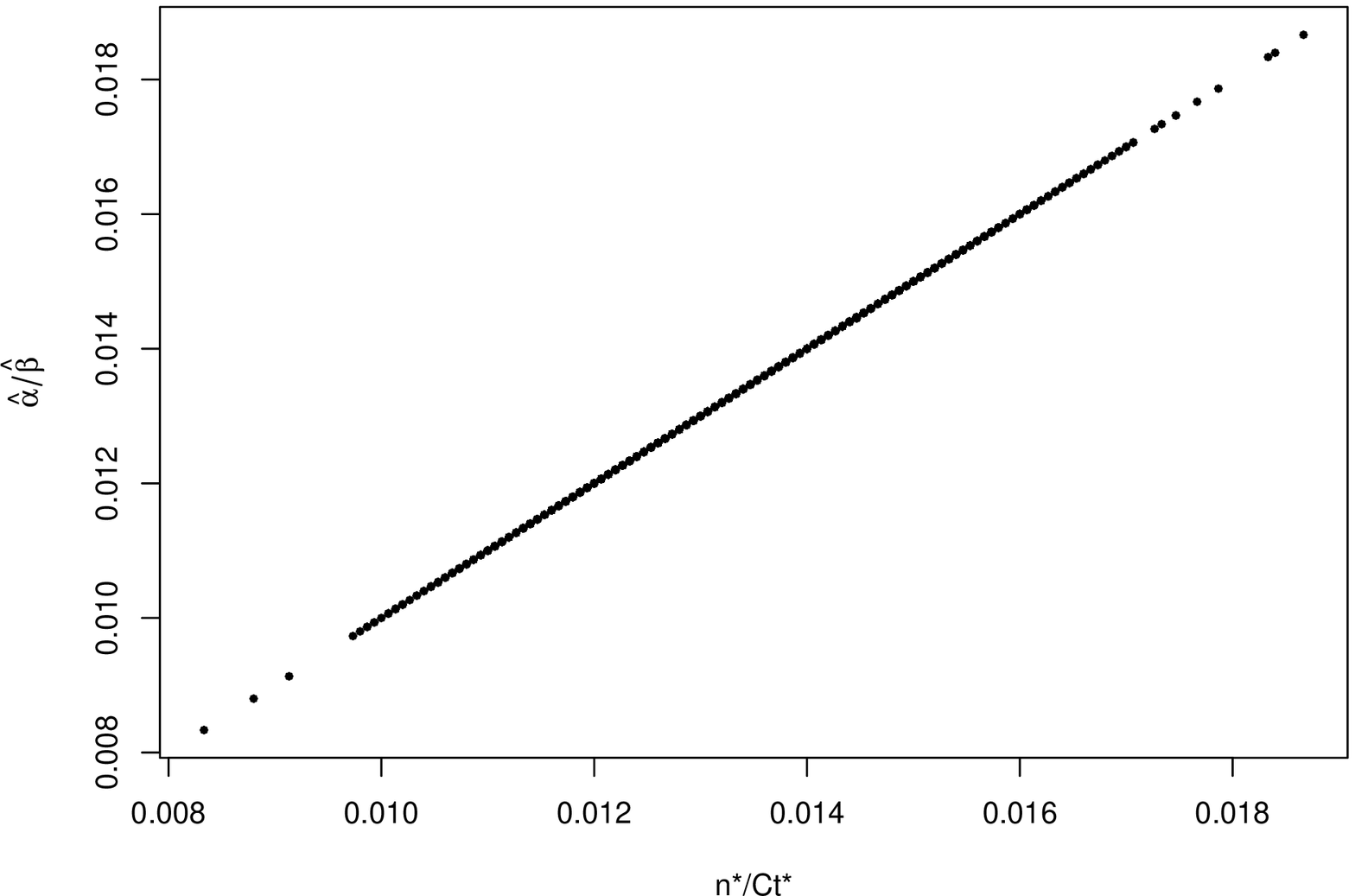}
    \end{tabular}
    \caption{Plot of $\hat{\alpha}/\hat{\beta}$ against $n^{*}/Ct^{*}$ for centre opening time scenario (1) (left) and scenario (2) (right) with $\alpha = 2$, $\beta=150$, $C=150$ and $t=200$.}
    \label{fig:abhat_nct}
\end{figure}

\FloatBarrier 

\section{Additional results for Section  \ref{sec.simulate}}

\subsection{Additional results for Section \ref{sec.sim.thrm}}
\label{sec.addn.simthrm}
Figure \ref{fig:PVI_t} considers objective (2). The census time is varied across simulations, while $n_\bullet^+$, the number of centres, and the individual centre recruitment rate are fixed. It  provides further validation of Theorem \ref{thrm.asympt}. The accuracy of the $p=0.25$ quantile is primarily dependent on the ratio of $n_\bullet^{+}/n_\bullet$, hence for a fixed $n_\bullet^{+}$, the density concentrates at the point mass $p$ with increasing census time. The observed effect of the census time on the accuracy of the predicted quantile compares well with the theoretical densities. 

\begin{figure}[ht]
    \centering
    \includegraphics[scale = 0.53]{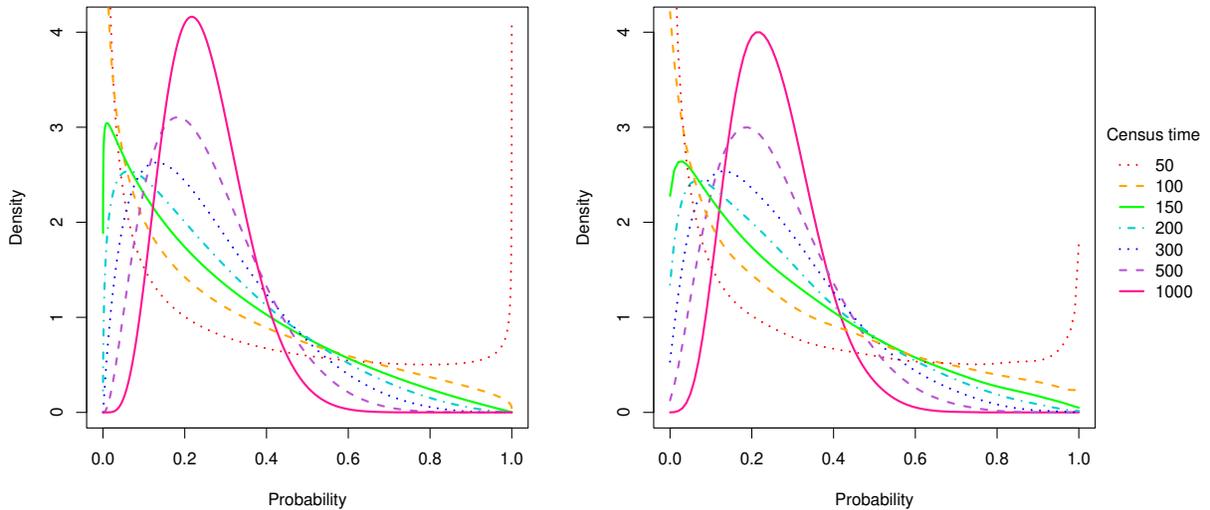}
    \caption{Theoretical density (left) and estimated density over repeated sampling (right) of $\Prob{T^{+} \leq \rhat_{0.25}}$ for each $t \in \mathbb{T}_{2}$ with $n_\bullet^{+}=200$ fixed across all simulation runs.}
    \label{fig:PVI_t}
\end{figure}

Figure \ref{fig:NB_fixedb} shows an equivalent plot to the right-hand side of Figure \ref{fig:change_t_m} of the main article, but with a fixed $\beta=150$. It gives the same patterns, which we would expect because we have not changed the ratio of $t^{+}/t$.

Figure \ref{fig:PVI_fixedb} shows an equivalent plot to Figure \ref{fig:PVI_C} of the main article, but now the rate parameter, $\beta$, is fixed and so the enrolment rate varies with number of centres. Here we see a very different pattern since $n$ is increasing with $C$, thus decreasing the ratio $n^{+}/n$ and improving model predictions.

\begin{figure}[ht]
    \centering
    \includegraphics[scale = 0.6]{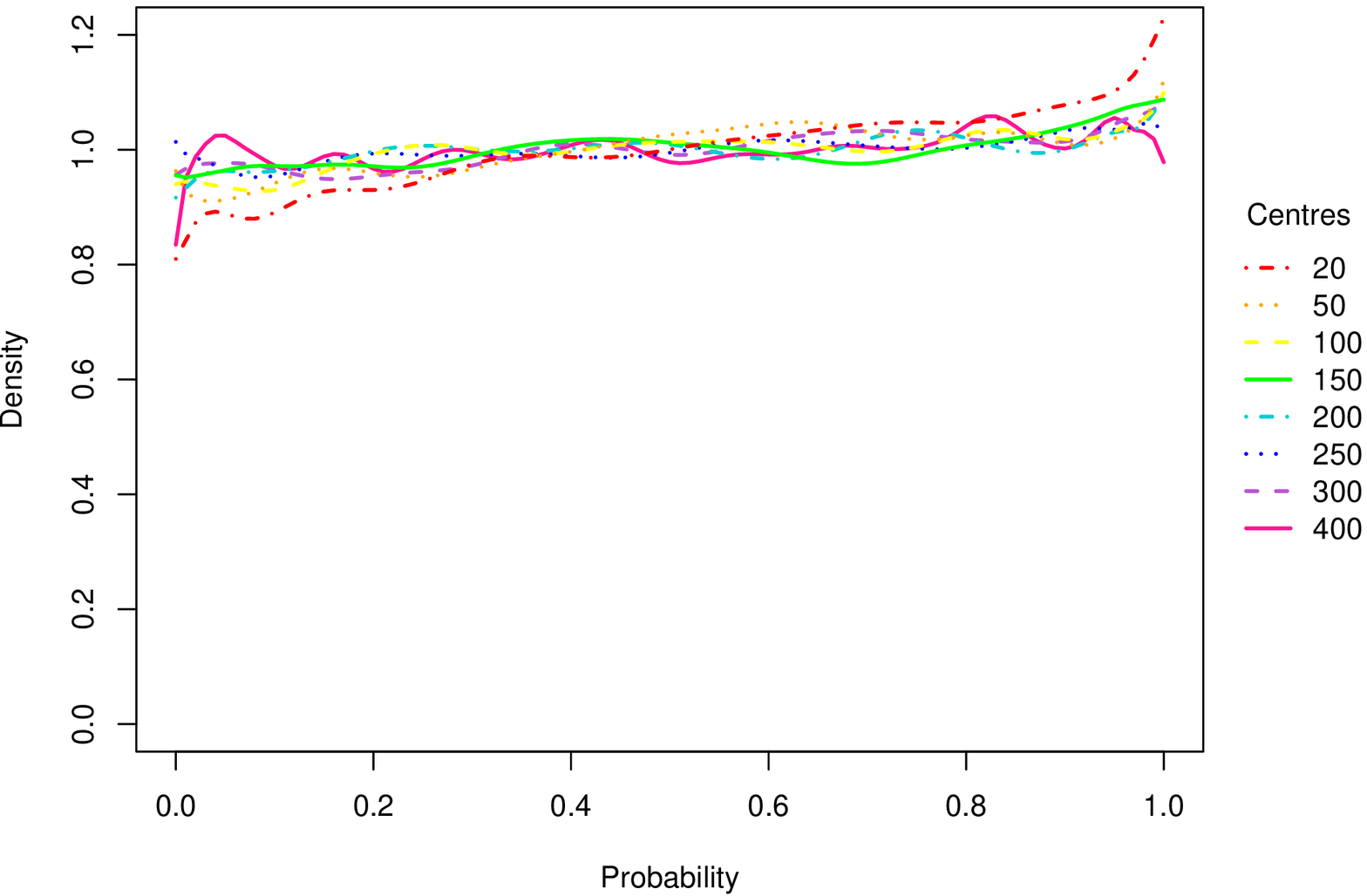}
    \caption{Estimated density (over repeated sampling) of $\Prob{N_\bullet^+\le \qhat_{0.5}}$ for each $C \in \boldsymbol{\mathbb{C}_{1}}$ with $t=t^{+}=200$ and $\beta=150$ fixed across all simulation runs.}
    \label{fig:NB_fixedb}
\end{figure}

\begin{figure}[ht]
    \centering
    \includegraphics[scale = 0.55]{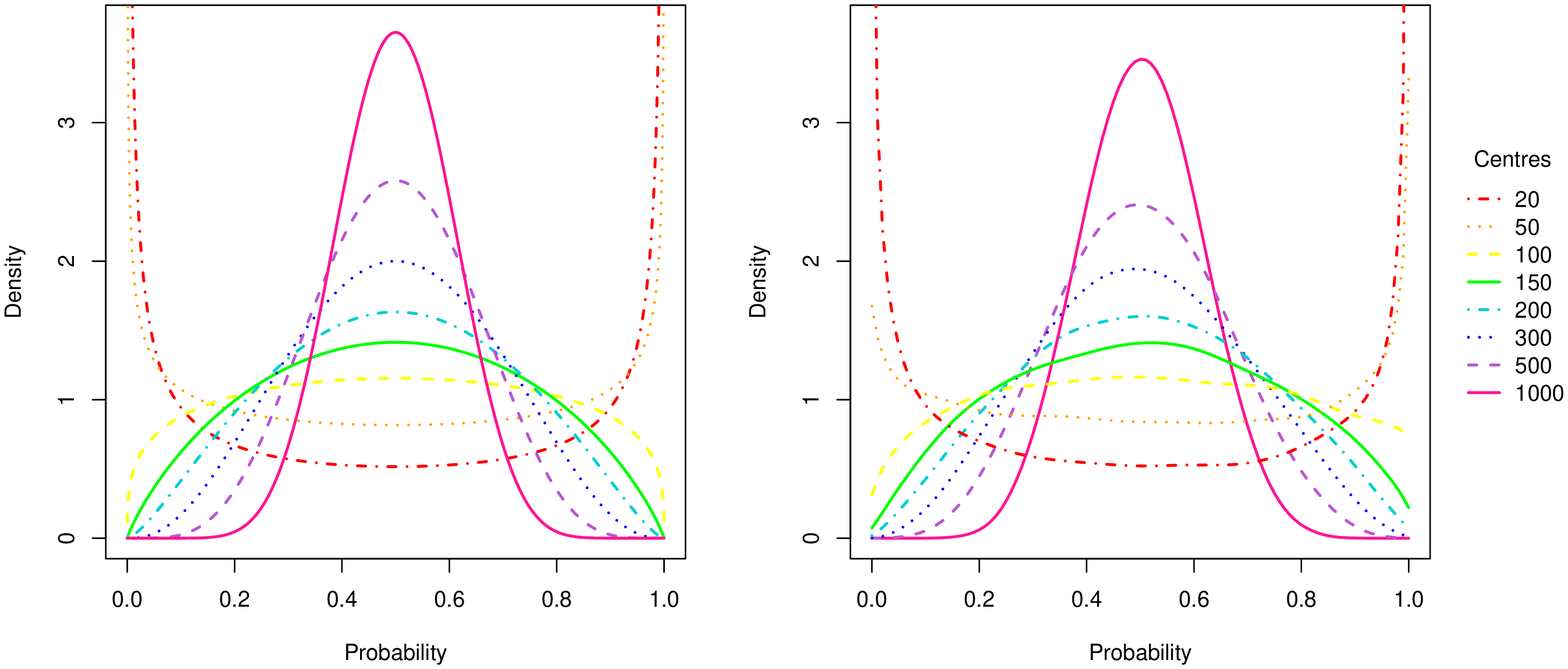}
    \caption{Theoretical density (left) and estimated density over repeated sampling (right) of $\Prob{T^{+} \leq \rhat_{0.5}}$ for each $C \in \boldsymbol{\mathbb{C}_{2}}$ with $n^{+}_\bullet=200$ and $\beta=150$ fixed across all simulation runs.}
    \label{fig:PVI_fixedb}
\end{figure}

\FloatBarrier

\subsection{Additional results for Section \ref{sec.sim.intervals}} \label{sec.addn.sim.intervals}

The results tables in this section evidence further investigation into the interval adjustment methodology.

Table \ref{tab:betafifty} shows that the methodology is still helpful for creating prediction intervals for $N_\bullet^+$ when $\beta$ is $50$ rather than $150$. Tables \ref{tab:quantile_adj_Ctwenty} and \ref{tab:quantile_adj_ninefive} correspond to $\beta=150$ but, respectively examining $95\%$ intervals or $90\%$ intervals with $C=20$.
The remaining tables display results of interval adjustment for $T^{+}$ for each of the three centre opening time scenarios considered.

\begin{table}[ht]
\centering
\caption{The mean (over repeated sampling) of the true coverage probability and width of an intended 90\% prediction interval for $N_\bullet^{+}$ with $\beta=50$ using the unadjusted and adjusted method. \vspace{1em}}
\label{tab:betafifty}
\begin{tabular}{@{}ccccc@{}}
\toprule
 & \multicolumn{2}{c}{\textbf{Unadjusted}} & \multicolumn{2}{c}{\textbf{Adjusted}} \\ \midrule
 & Coverage  (\%) & \textbf{$w$} & Coverage  (\%) & \textbf{$w$} \\ \midrule
$t=50, t^{+}=350$ & 77.8 & 317.1 & 90.2 & 426.5 \\ \midrule
$t=100, t^{+}=300$ & 84.8 & 240.7 & 90.2 & 278.8 \\ \midrule
$t=150, t^{+}=250$ & 86.2 & 190.7 & 89.4 & 207.9 \\ \midrule
\multicolumn{1}{l}{$t=200, t^{+}=200$} & 88.3 & 152.7 & 90.2 & 161.2 \\ \midrule
\multicolumn{1}{l}{$t=250, t^{+}=150$} & 88.7 & 120.8 & 89.8 & 124.8 \\ \midrule
\multicolumn{1}{l}{$t=300, t^{+}=100$} & 89.6 & 91.3 & 90.2 & 93.0 \\ \midrule
\multicolumn{1}{l}{$t=350, t^{+}=50$} & 89.8 & 60.5 & 90.1 & 60.9 \\ 
\bottomrule
\end{tabular}
\end{table}


\begin{table}[ht]
\centering
\caption{The mean (over repeated sampling) of the true coverage probability and width of an intended 90\% prediction interval for $N_\bullet^{+}$ with $C=20$ using the unadjusted and adjusted method. \vspace{1em}}
\label{tab:quantile_adj_Ctwenty}
\begin{tabular}{@{}ccccc@{}}
\toprule
 & \multicolumn{2}{c}{\textbf{Unadjusted}} & \multicolumn{2}{c}{\textbf{Adjusted}} \\ \midrule
 & Coverage  (\%) & $w$ & Coverage  (\%) & $w$ \\ \midrule
$t=50, t^{+}=350$ & 59.2 & 46.5 & 89.7 & 88.3 \\ \midrule
$t=100, t^{+}=300$ & 73.4 & 40.2 & 89.9 & 58.0 \\ \midrule
$t=150, t^{+}=250$ & 80.3 & 34.5 & 89.9 & 43.4 \\ \midrule
\multicolumn{1}{l}{$t=200, t^{+}=200$} & 84.1 & 29.1 & 90.0 & 33.7 \\ \midrule
\multicolumn{1}{l}{$t=250, t^{+}=150$} & 86.5 & 23.8 & 90.0 & 26.0 \\ \midrule
\multicolumn{1}{l}{$t=300, t^{+}=100$} & 87.8 & 18.4 & 89.8 & 19.4 \\ \midrule
\multicolumn{1}{l}{$t=350, t^{+}=50$} & 88.4 & 12.4 & 89.3 & 12.7 \\ \bottomrule
\end{tabular}
\end{table}

\FloatBarrier

\begin{table}[ht]
\centering
\caption{The mean (over repeated sampling) of the true coverage probability and width of an intended 95\% prediction interval for $N^{+}_\bullet$ using the unadjusted and adjusted method. \vspace{0.5em}}
\label{tab:quantile_adj_ninefive}
\begin{tabular}{@{}ccccc@{}}
\toprule
 & \multicolumn{2}{c}{\textbf{Unadjusted}} & \multicolumn{2}{c}{\textbf{Adjusted}} \\ \midrule
 & Coverage  (\%) & \textbf{$w$} & Coverage  (\%) & \textbf{$w$} \\ \midrule
$t=50, t^{+}=350$ & 72.0 & 167.4 & 94.4 & 292.7 \\ \midrule
$t=100, t^{+}=300$ & 84.2 & 140.9 & 94.8 & 191.7 \\ \midrule
$t=150, t^{+}=250$ & 88.9 & 118.0 & 94.7 & 143.0 \\ \midrule
\multicolumn{1}{l}{$t=200, t^{+}=200$} & 91.3 & 97.9 & 94.7 & 110.7 \\ \midrule
\multicolumn{1}{l}{$t=250, t^{+}=150$} & 92.8 & 79.4 & 94.8 & 85.8 \\ \midrule
\multicolumn{1}{l}{$t=300, t^{+}=100$} & 93.8 & 61.2 & 94.9 & 63.9 \\ \midrule
\multicolumn{1}{l}{$t=350, t^{+}=50$} & 94.5 & 41.1 & 94.9 & 41.9 \\ \bottomrule
\end{tabular}
\end{table}

\FloatBarrier

\begin{table}[ht]
\centering
\caption{The mean (over repeated sampling) of the true coverage probability and width of an intended 90\% prediction interval for $T^{+}$ with $n_\bullet^{+}=200$ using the unadjusted and adjusted method.\vspace{1em}}
\label{tab:quantile_adj_PVI}
\begin{tabular}{@{}cccccc@{}}
\toprule
 & \multicolumn{2}{c}{\textbf{Unadjusted}} & \multicolumn{2}{c}{\textbf{Adjusted}} & \\ \midrule
 & Coverage  (\%) & $w$ & Coverage  (\%) & $w$ & \\ \midrule
$t=50$ & 73.9 & 28.7 & 89.6 & 41.5 & \\ \midrule
$t=100$ & 82.4 & 27.7 & 89.7 & 33.4 &\\ \midrule
$t=150$ & 85.4 & 27.0 & 89.7 & 30.4 &\\ \midrule
$t=200$ & 86.8 & 26.5 & 89.7 & 28.8 &\\ \midrule
$t=300$ & 88.2 & 25.9 & 89.8 & 27.1 &\\ \midrule
$t=500$ & 89.4 & 25.1 & 90.1 & 25.6 &\\ \midrule
$t=1000$ & 89.8 & 24.4 & 90.0 & 24.5 &\\ \bottomrule
\end{tabular}
\end{table}

 \FloatBarrier

\begin{table}[ht]
\centering
\caption{The mean (over repeated sampling) true coverage probability and width  of an intended 90\% prediction interval for $T^{+}$ with $n_\bullet^{+}=200$ using the unadjusted and adjusted method for opening time scenario (1).\vspace{1em}}
\label{tab:PVI_diffc}
\begin{tabular}{@{}ccccccccc@{}}
\toprule 
 & \multicolumn{1}{l}{} & \multicolumn{1}{l}{} & \multicolumn{1}{l}{} & \multicolumn{2}{c}{\textbf{Unadjusted}} & \multicolumn{2}{c}{\textbf{Adjusted}} &\\ \midrule
 & \multicolumn{1}{c}{$t^{*}$} & \multicolumn{1}{c}{$t^{*} / t_{c}$} & \multicolumn{1}{c}{$n_{\bullet}^{*} / n_{\bullet} $} & Coverage (\%) & $w$ & Coverage (\%) & $w$ &\\ \midrule
$t=50$ & 24.4 & 0.957 & 0.956 & 61.6 & 29.2 & 89.8 & 55.1 &\\ \midrule
$t=100$ & 46.5 & 0.921 & 0.920 & 73.7 & 28.7 & 89.9 & 42.4 &\\ \midrule
$t=150$ & 67.2 & 0.891 & 0.890 & 78.7 & 28.3 & 89.9 & 37.4 &\\ \midrule
$t=200$ & 86.9 & 0.866 & 0.866 & 81.5 & 27.9 & 89.9 & 34.6 &\\ \midrule
$t=300$ & 124.2 & 0.826 & 0.825 & 84.4 & 27.3 & 89.7 & 31.6 &\\ \midrule
$t=500$ & 192.5 & 0.769 & 0.768 & 86.8 & 26.6 & 89.9 & 28.9 &\\ \midrule
$t=1000$ & 344.3 & 0.689 & 0.687 & 88.6 & 25.7 & 89.9 & 26.6 &\\ \bottomrule
\end{tabular}
\end{table}

\FloatBarrier

\begin{table}[ht]
\centering
\caption{The mean (over repeated sampling) true coverage probability and width of an intended 90\% prediction interval for $T^{+}$ with $n_\bullet^{+}=200$ using the unadjusted and adjusted method for opening time scenario (2).\vspace{1em}}
\label{tab:PVI_diffc_half}
\begin{tabular}{@{}ccccccccc@{}}
\toprule
 & \multicolumn{1}{l}{} & \multicolumn{1}{l}{} & \multicolumn{1}{l}{} & \multicolumn{2}{c}{\textbf{Unadjusted}} & \multicolumn{2}{c}{\textbf{Adjusted}} & \\ \midrule
 & \multicolumn{1}{c}{$t^{*}$} & \multicolumn{1}{c}{$t^{*} / t_{c}$} & \multicolumn{1}{c}{$n_{\bullet}^{*} / n_{\bullet} $} & Coverage (\%) & $w$ & Coverage (\%) & $w$ & \\ \midrule
$t=50$ & 21.7 & 0.867 & 0.863 & 60.1 & 29.9 & 89.7 & 59.1 &\\ \midrule
$t=100$ & 38.3 & 0.766 & 0.763 & 69.7 & 29.2 & 89.5 & 46.1 &\\ \midrule
$t=150$ & 51.2 & 0.683 & 0.679 & 74.2 & 28.9 & 89.4 & 41.3 &\\ \midrule
$t=200$ & 61.7 & 0.617 & 0.614 & 76.5 & 28.6 & 89.3 & 38.8 &\\ \midrule
$t=300$ & 77.0 & 0.513 & 0.512 & 79.6 & 28.3 & 89.5 & 36.1 &\\ \midrule
$t=500$ & 96.1 & 0.384 & 0.382 & 82.2 & 27.9 & 89.7 & 33.8 &\\ \midrule
$t=1000$ & 118.9 & 0.238 & 0.236 & 83.9 & 27.5 & 89.7 & 32.1 &\\ \bottomrule
\end{tabular}
\end{table}

\FloatBarrier

\subsection{Diagnostics for Section \ref{sec.AZdata}}
\label{sec.addn.diagnostics}
In this section, we present diagnostics which show the suitability the Poisson-Gamma model for the oncology clinical trial data analysed in Section \ref{sec.AZdata}.

According to the model, the marginal distribution of the counts in a given initial period $[0,a]$ of each centre's recruitment is negative binomial. Each centre's recruitment count, in theory, follows the same distribution, which depends on $\alpha$ and $\beta$ only, and they are independent from each other. A QQ-plot can be used as a diagnostic test for this assumption, by comparing the quantiles of the theoretical negative binomial distribution with those of the observed recruitment counts for the individual centres. 

Figure \ref{fig:QQplot} shows two Q-Q plots, applied to the oncology clinical trial data set used in Section \ref{sec.AZdata} in the main article. The first uses $a=t/2=0.0625$, and considers only the data that would have been available at the census time, $t=0.125$ . Thus it looks only at the 18 centres that had already opened by time $a=0.0625$.   To verify that the negative binomial distribution is reasonable for all of the centres, the second Q-Q plot sets $a=\tau/10=0.1$, and uses all of the available information.  The plots are both close to a straight line, suggesting that the assumption of a hierarchical gamma distribution for the centre intensities is reasonable.

\begin{figure}[ht]
    \centering
    \subfigure[a = 0.0625]{\includegraphics[width = 0.48 \textwidth, height = 7.5cm]{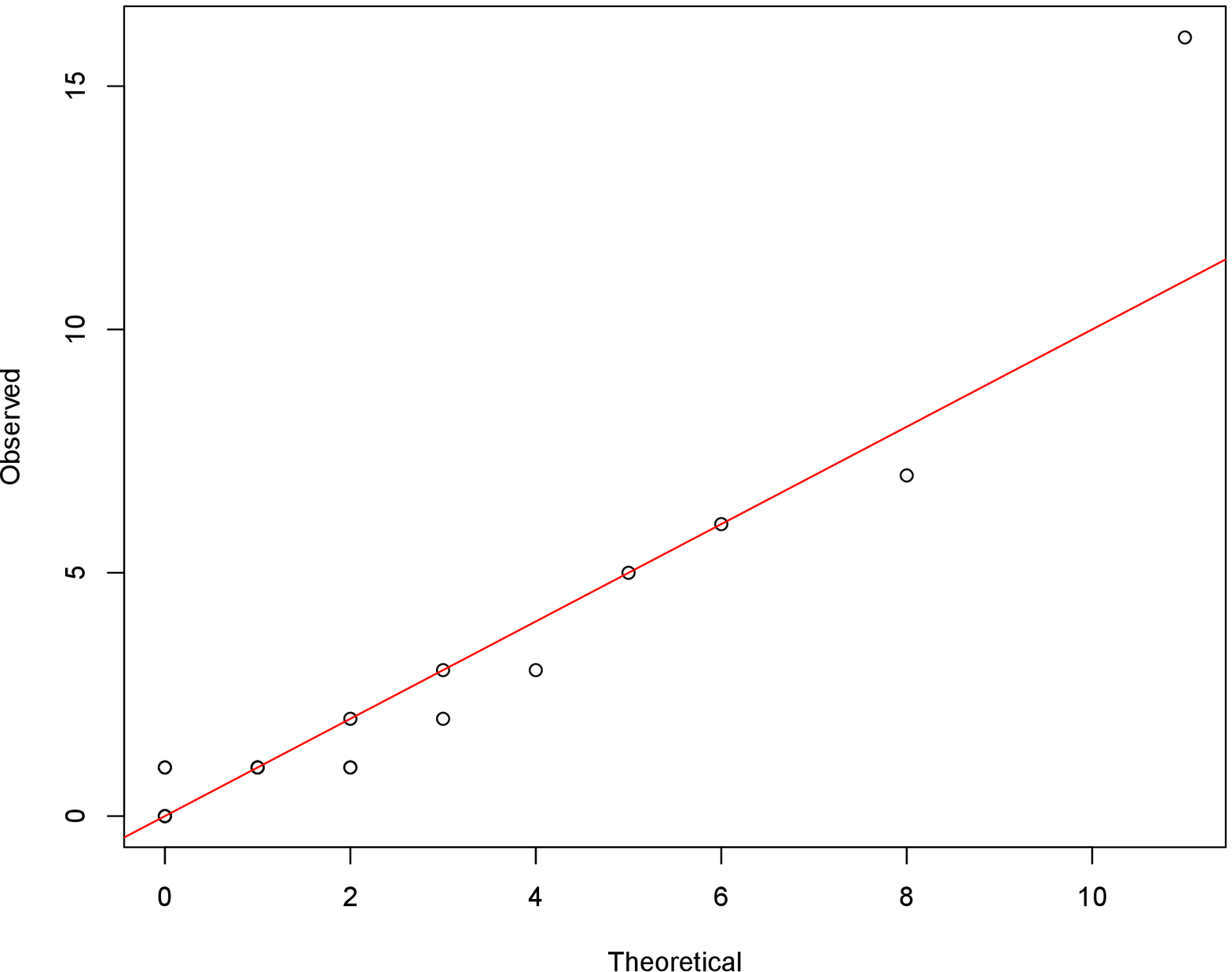}}
    \hspace{1em}
    \subfigure[a = 0.1]{\includegraphics[width = 0.48 \textwidth, height = 7.5cm]{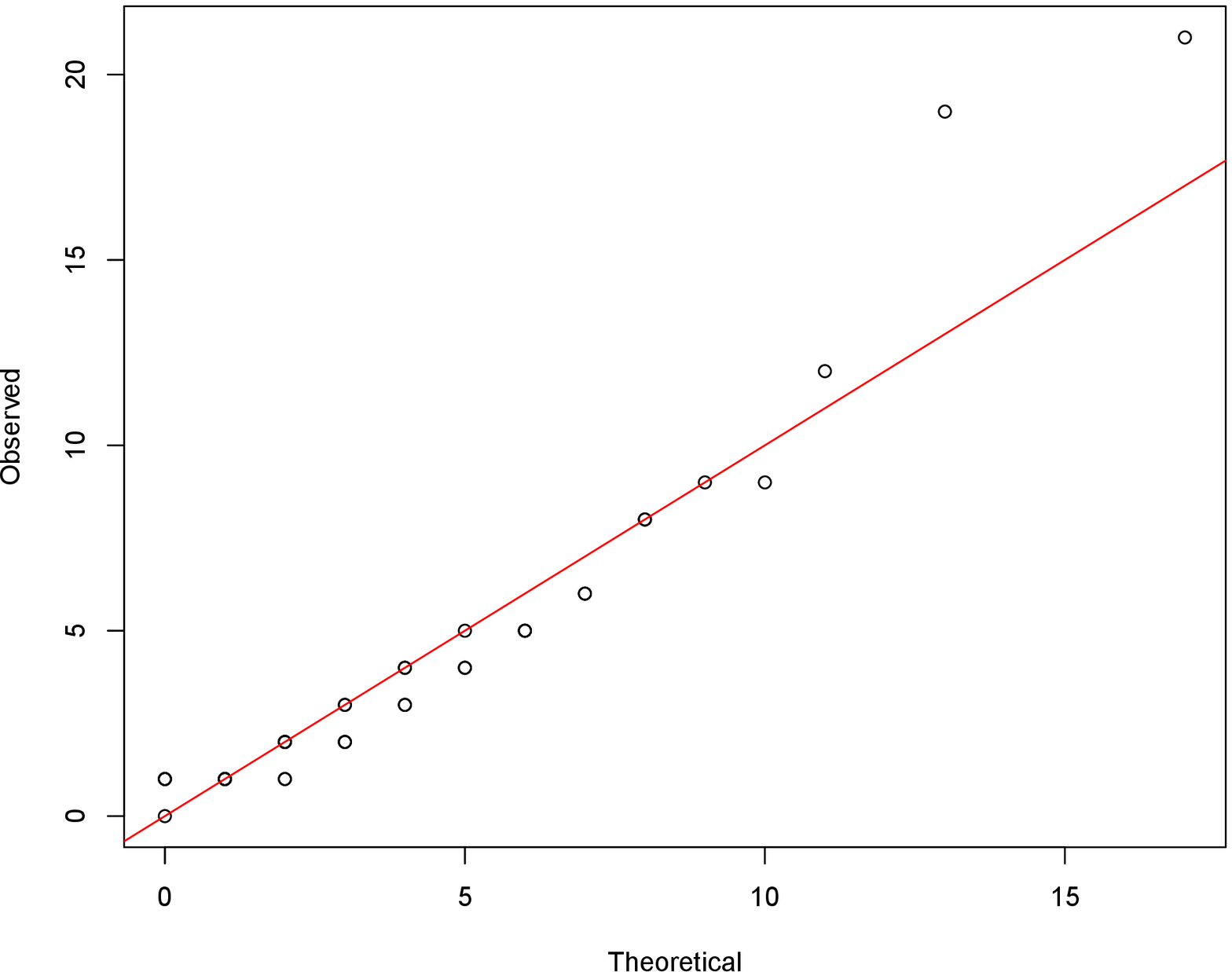}}
    \caption{QQ-plots for observed recruitment to each centre of an oncology multicentre clinical trial compared to the theoretical negative binomial distribution.}
    \label{fig:QQplot}
\end{figure}

\FloatBarrier

\section{Departure from the hierarchical gamma distribution}
\label{gam.mix}

This section investigates the impact of departure from the assumption of a gamma hierarchical distribution on the interval adjustment methodology.

We simulate from a scenario where the individual centre recruitment rates, $\lambda_{c}$, come from a mixture of two Gamma distributions with a density of
\begin{align*}
    f(\lambda;\alpha,\beta_1,\beta_2)&=
    \frac{1}{2}
    \mbox{Gam}(\lambda;\alpha, \beta_{1})
    +
    \frac{1}{2}
    \mbox{Gam}(\lambda;\alpha, \beta_{2}).
\end{align*} 
Table \ref{tab:gamma_mix_NB} displays results of the interval adjustment for $N_\bullet^+$ with $\beta_{1}=150$ and $\beta_{2}=450$, and uniformly distributed centre opening times. Table \ref{tab:gamma_mix_PVI} gives analogous results for $T^+$. Both investigations show that despite the incorrect specification of the hierarchical distribution in the fitted model, the coverage of the adjusted prediction intervals is a considerable improvement on the unadjusted interval and is very close to the intended coverage.

\begin{table}[ht]
\centering
\caption{The mean (over repeated sampling) coverage and width of an intended 90\% prediction interval for $N^{+}_\bullet$ with $\beta_{1}=150$ and $\beta_{2}=450$ using the unadjusted and adjusted method for opening time scenario (1).\vspace{1em}}
\label{tab:gamma_mix_NB}
\begin{tabular}{@{}ccccccccc@{}}
\toprule
 & \multicolumn{1}{l}{} & \multicolumn{1}{l}{} & \multicolumn{1}{l}{} & \multicolumn{2}{c}{\textbf{Unadjusted}} & \multicolumn{2}{c}{\textbf{Adjusted}} & \\ \midrule
 & \multicolumn{1}{c}{$t^{*}$} & \multicolumn{1}{c}{$t^{*} / t_{c}$} & \multicolumn{1}{c}{$n_{\bullet}^{*} / n_{\bullet} $} & Coverage (\%) & $w$ & Coverage (\%) & $w$ &\\ \midrule
$t=50, t^{+}=350$ & 24.2 & 0.948 & 0.207 & 52.2 & 122.7 & 89.9 & 278.0 &\\ \midrule
$t=100, t^{+}=300$ & 46.1 & 0.913 & 0.911 & 65.5 & 105.3 & 89.7 & 180.0 &\\ \midrule
$t=150, t^{+}=250$ & 66.6 & 0.882 & 0.880 & 73.6 & 89.1 & 89.9 & 130.9 &\\ \midrule
$t=200, t^{+}=200$ & 86.1 & 0.856 & 0.854 & 78.4 & 73.7 & 90.0 & 97.9 &\\ \midrule
$t=250, t^{+}=150$ & 104.8 & 0.835 & 0.833 & 81.8 & 59.1 & 89.8 & 72.5 &\\ \midrule
$t=300, t^{+}=100$ & 122.9 & 0.816 & 0.814 & 84.3 & 44.6 & 89.5 & 51.2 &\\ \midrule
$t=350, t^{+}=50$ & 140.3 & 0.799 & 0.797 & 87.1 & 29.2 & 89.6 & 31.3 &\\ \bottomrule
\end{tabular}
\end{table}

\begin{table}[ht]
\centering
\caption{The mean (over repeated sampling) coverage and width of an intended 90\% prediction interval for $T^{+}$ with $n_\bullet^{+}=200$, $\beta_{1}=150$, $\beta_{2}=450$ using the unadjusted and adjusted method for opening time scenario (2).\vspace{1em}}
\label{tab:gamma_mix_PVI}
\begin{tabular}{@{}ccccccccc@{}}
\toprule
 & \multicolumn{1}{l}{} & \multicolumn{1}{l}{} & \multicolumn{1}{l}{} & \multicolumn{2}{c}{\textbf{Unadjusted}} & \multicolumn{2}{c}{\textbf{Adjusted}} &\\ \midrule
 & \multicolumn{1}{c}{$t^{*}$} & \multicolumn{1}{c}{$t^{*} / t_{c}$} & \multicolumn{1}{c}{$n_{\bullet}^{*} / n_{\bullet} $} & Coverage (\%) & $w$ & Coverage (\%) & $w$ &\\ \midrule
$t=50$ & 24.2 & 0.949 & 0.948 & 59.3 & 50.3 & 90.6 & 102.7 &\\ \midrule
$t=100$ & 46.1 & 0.913 & 0.911 & 70.5 & 48.0 & 90.0 & 75.6 &\\ \midrule
$t=150$ & 66.6 & 0.882 & 0.880 & 76.6 & 46.8 & 90.0 & 65.1 &\\ \midrule
$t=200$ & 86.1 & 0.856 & 0.854 & 79.9 & 45.8 & 90.1 & 59.4 &\\ \midrule
$t=300$ & 122.9 & 0.816 & 0.814 & 82.6 & 44.3 & 89.5 & 53.0 &\\ \midrule
$t=500$ & 190.4 & 0.759 & 0.758 & 86.0 & 42.5 & 89.9 & 47.4 &\\ \midrule
$t=1000$ & 340.2 & 0.679 & 0.679 & 88.2 & 40.3 & 89.9 & 42.3 &\\ \bottomrule
\end{tabular}
\end{table}

\FloatBarrier

\end{document}